\documentclass{article}
\usepackage[margin=1in]{geometry}
\usepackage{layout}
\usepackage{xcolor}
\usepackage{amssymb}
\usepackage{amsthm}
\newtheorem{lemma}{Lemma}[section]
\newtheorem{corollary}{Corollary}[section]
\newtheorem{theorem}{Theorem}[section]
\newtheorem{proposition}{Proposition}[section]
\theoremstyle{definition} 
\newtheorem{definition}{Definition}[section]
\newtheorem{example}{Example}[section]

\usepackage[auth-sc-lg]{authblk}
\author{Aziz Amezian El Khalfioui\thanks{Research fellow FNRS.}
}
\author{
Jef Wijsen
}
\affil{University of Mons, Mons, Belgium\\
\texttt{\{aziz.amezianelkhalfioui, jef.wijsen\}@umons.ac.be}}
\date{}

\usepackage{complexity}
\usepackage{xspace}	
\usepackage{bigstrut}
\usepackage{arydshln}
\usepackage{enumitem}

\usepackage{mathtools}
\newcommand{\myheading}[1]{\vspace{0.5\baselineskip}{\sc{#1}}}
\newcommand{\defeq}{\vcentcolon=}
\newcommand{\card}[1]{|{#1}|}
\newcommand{\arity}[1]{|{#1}|}
\newcommand{\db}{\mathbf{db}}
\newcommand{\tuple}[1]{\langle{#1}\rangle}
\newcommand{\makefree}[2]{\nexists{#2}\left[{#1}\right]}
\newcommand{\makebound}[2]{\exists{#2}\left[{#1}\right]}
\newcommand{\substitute}[3]{{#1}_{[{#2}\rightarrow{#3}]}}
\newcommand{\cqamodels}{\models_{\mathsf{\tiny cqa}}}
\newcommand{\attribute}[1]{\mathit{#1}}
\newcommand{\constant}[1]{\textnormal{`{#1}'}}
\newcommand{\mycount}[2]{{\mathbf{cnt}}({#1},{#2})}
\newcommand{\cqacount}[2]{{\mathbf{cqacnt}}({#1},{#2})}
\newcommand{\formula}[1]{\left({#1}\right)}
\newcommand{\signature}[2]{[{#1},{#2}]}
\newcommand{\cforest}{\mathsf{Cforest}}
\newcommand{\cnewclass}{\mathsf{Cparsimony}}
\newcommand{\rep}{\mathbf{r}}
\newcommand{\sep}{\mathbf{s}}
\newcommand{\oep}{\mathbf{o}}
\newcommand{\pep}{\mathbf{p}}
\newcommand{\calD}{\mathcal{D}}
\newcommand{\calR}{\mathcal{R}}
\newcommand{\fd}[2]{#1 \rightarrow #2}
\newcommand{\fdset}[1]{\mathcal{K}(#1)}
\newcommand{\vars}[1]{{\mathsf{vars}}(#1)}
\newcommand{\free}[1]{{\mathsf{free}}(#1)}
\newcommand{\key}[1]{{\mathsf{Key}}(#1)}
\newcommand{\frozen}[1]{{\mathsf{frozen}}(#1)}
\newcommand{\notkey}[1]{{\mathsf{notKey}}(#1)}
\newcommand{\body}[1]{{\mathsf{body}}(#1)}
\newcommand{\keycl}[2]{{#1}^{+,#2}}
\newcommand{\attacksrelation}[1]{\stackrel{#1}{\rightsquigarrow}}
\newcommand{\nattacksrelation}[1]{\stackrel{#1}{\not\rightsquigarrow}}
\newcommand{\attacks}[3]{{#1}\attacksrelation{#3}{#2}}
\newcommand{\notattacks}[3]{{#1}\nattacksrelation{#3}{#2}}
\newcommand{\cqa}[1]{\mathsf{CERTAINTY}({#1})}
\newcommand{\problem}[1]{{\textsf{#1}}}

\title{Consistent Query Answering for Primary Keys and Conjunctive Queries with Counting} 

\begin{document}

\maketitle

\begin{abstract}
The problem of consistent query answering for primary keys and self-join-free conjunctive queries has been intensively studied in recent years and is by now well understood.
In this paper, we study an extension of this problem with counting.
The queries we consider count how many times each value occurs in a designated (possibly composite) column of an answer to a full conjunctive query. 
In a setting of database repairs, we adopt the semantics of [Arenas et al., ICDT 2001] which computes tight lower and upper bounds on these counts, where the bounds are taken over all repairs.
Ariel Fuxman defined in his PhD thesis a syntactic class of queries, called $\cforest$, for which this computation can be done by executing two first-order queries (one for lower bounds, and one for upper bounds) followed by simple counting steps. 
We use the term ``parsimonious counting'' for this computation.
A natural question is whether $\cforest$ contains all self-join-free conjunctive queries that admit parsimonious counting. 
We answer this question negatively.
We define a new syntactic class of queries, called $\cnewclass$, and prove that it contains all (and only) self-join-free conjunctive queries that admit parsimonious counting.
\end{abstract}

\section{Introduction}\label{sec:introduction}

The problem of consistent query answering (CQA)~\cite{DBLP:conf/icdt/ArenasBC01,DBLP:series/synthesis/2011Bertossi,DBLP:conf/pods/Bertossi19,DBLP:journals/sigmod/Wijsen19} with respect to primary keys is by now well understood for self-join-free conjunctive queries:
a dichotomy between tractable and intractable queries has been established, 
and it is known which queries have a consistent first-order rewriting~\cite{DBLP:journals/tods/KoutrisW17,DBLP:journals/mst/KoutrisW21}. 
It remains a largely open question to extend these complexity results to queries with aggregation.
In this paper, we look at a simple form of aggregation: counting the number of times each (possibly composite) value occurs in the answer to a conjunctive query.
Although this problem has been studied since the early years of CQA~\cite{FuxmanThesis}, a fine-grained characterization of its complexity remains open.

Formally, let $q$ be a full (i.e., quantifier-free) self-join-free conjunctive query.
We define a counting query as follows.
We designate a tuple $\vec{z}$ of distinct variables of $q$, called the \emph{grouping variables},
and let $\vec{w}$ be a tuple of the variables in~$q$ that are not in~$\vec{z}$.
The variables of~$q$, which are all free, are made explicit by denoting $q$ as $q(\vec{z},\vec{w})$.
We are interested in a query that, on a given database instance~$\db$, returns all tuples $(\vec{c},i)$ with $\vec{c}$ a tuple of constants, of the same arity as~$\vec{z}$, and with $i$ a positive integer that is the number of distinct tuples $\vec{d}$, of the same arity as $\vec{w}$, satisfying $(\vec{c},\vec{d})\in q(\db)$. 
This counting query will be denoted $\mycount{q}{\vec{z}}$.

For example, consider the database schema of Fig.~\ref{fig:exampledatabase}, which is intended to store the unique gender and department of each employee, and the unique building of each department.
Ignore for now that the database instance of Fig.~\ref{fig:exampledatabase} is inconsistent (as it stores two departments for Anny, and two buildings for IT). 
Let $q_{0}(x,y,z)=E(x,\constant{F},y)\land D(y,z)$, where $x,y,z$ are variables and $\constant{F}$ denotes a constant.
Then, on a consistent database instance, $\mycount{q_{0}}{z}$ would return the number of female employees working in each building.
In SQL, $\mycount{q_{0}}{z}$ can be encoded as follows:
\begin{quote}\small
\begin{verbatim}
SELECT   Building, COUNT(*) AS CNT
FROM     E, D
WHERE    E.Dept = D.Dept AND Gender = 'F'
GROUP BY Building
\end{verbatim}
\end{quote}
On the database instance of Fig.~\ref{fig:exampledatabase}, this query will return $(A,4)$ and $(B,3)$.
These answers are however not meaningful because they suffer from double-counting due to inconsistencies.
We describe next a more meaningful semantics that was introduced in~\cite{DBLP:conf/icdt/ArenasBC01}.

First, following~\cite{DBLP:conf/pods/ArenasBC99}, we define a \emph{repair} of a database instance as a maximal subinstance that satisfies all primary-key constraints.
In this paper, we consider no other constraints than primary keys.
Then, following the approach of~\cite{DBLP:conf/icdt/ArenasBC01}, more meaningful answers are obtained by returning, for every value $\vec{c}$ for the grouping variables $\vec{z}$, tight lower and upper bounds on the corresponding counts over all repairs.
This new query is denoted by $\cqacount{q}{\vec{z}}$.
Thus, an answer $(\vec{c},[m,n])$ to this new query means that on every repair, our original query $\mycount{q}{\vec{z}}$ returns a tuple $(\vec{c},i)$ with $m\leq i\leq n$, and, moreover, these bounds~$m$ and~$n$ are tight.
By tight, we mean that for every $j\in\{m,n\}$, there is a repair on which $\mycount{q}{\vec{z}}$ returns $(\vec{c},j)$.

For example, the database instance of Fig.~\ref{fig:exampledatabase} has four repairs, because there are two choices for Anny's department, and two choices for the building of the IT department.
Note that in Fig.~\ref{fig:exampledatabase}, \emph{blocks} of conflicting tuples are separated by dashed lines. 
The query $\mycount{q_{0}}{z}$ returns different answers on each repair:
there are two repairs where the answer is $\{(A,3),(B,1)\}$; there is one repair where the answer is $\{(A,1),(B,3)\}$; and there is one repair where the answer is $\{(A,2),(B,2)\}$.
The latter set of answers, for example, is obtained in the repair that assigns Anny to department~HR, and~IT to building~B.
The query $\cqacount{q_{0}}{z}$ would thus return $\{(A,[1,3]), (B,[1,3])\}$.

In this paper, we are concerned about the complexity of computing $\cqacount{q}{\vec{z}}$.
In general, there exist self-join-free conjunctive queries $q$ such that, for some choice of the grouping variables $\vec{z}$, $\cqacount{q}{\vec{z}}$ cannot be solved in polynomial time (under standard complexity assumptions).
This follows from earlier research showing that there are self-join-free conjunctive queries~$q'(\vec{z})$ for which the following problem is $\coNP$-complete: given $\vec{c}$ and $\db$, determine whether $q'(\vec{c})$ is true in every repair of $\db$. 
The latter problem obviously reduces to counting: $q'(\vec{c})$ is true in every repair of $\db$ if and only if $\cqacount{q}{\vec{z}}$ returns $(\vec{c},[m,n])$ on $\db$ for some $m\geq 1$, where $q$ is the full query obtained from~$q'$ by dropping quantification.

In his PhD thesis~\cite{FuxmanThesis}, Fuxman showed that for some $q$ and $\vec{z}$, the answer to $\cqacount{q}{\vec{z}}$ can be computed by executing first-order queries followed by simple counting steps.
To illustrate his approach, consider the following query in SQL:
\begin{quote}\small
\begin{verbatim}
SELECT   Building, COUNT(DISTINCT Emp) AS CNT
FROM     E, D
WHERE    E.Dept = D.Dept AND Gender = 'F'
GROUP BY Building
\end{verbatim}
\end{quote}
On our example database of Fig.~\ref{fig:exampledatabase}, this query returns $\{(A,3), (B,3)\}$.
We observe that the returned counts match the upper bounds previously found for $\cqacount{q_{0}}{z}$.
Importantly, it can be shown that this is not by accident: on \emph{every} database instance, the latter SQL query will return the correct upper bounds for $\cqacount{q_{0}}{z}$.
Note that the latter query uses \texttt{COUNT(DISTINCT Emp)}, which means that duplicates are removed, which is a standard practice in relational algebra.

We now explain how to obtain the lower bounds for our example query.
To this end, consider the following query:
\begin{quote}\small
\begin{verbatim}
SELECT Building, Emp
FROM   E, D
WHERE  E.Dept = D.Dept AND Gender = 'F'
\end{verbatim}
\end{quote}
Following~\cite{DBLP:conf/pods/ArenasBC99}, we define the \emph{consistent answer} to such a query as the intersection of the query answers on all repairs.
For our example database, the consistent answer is the following table, which we call $C$:
\begin{quote}\small
\begin{tabular}{l|lc}
$C$ & $\attribute{Emp}$ & $\attribute{Building}$\bigstrut\\\cline{2-3}
& Suzy  & A\\
& Lucy & B
\end{tabular}
\end{quote}
Note that Anny does not occur in the consistent answer, because (Anny, A) is false in some repair, and so is (Anny, B).
From~\cite{DBLP:journals/tods/KoutrisW17}, it follows that computing the consistent answers to the latter SQL query is in $\FO$ (i.e., the class of problems that can be solved by a first-order query), using a technique known as \emph{consistent first-order rewriting}.
The lower bounds $\{(A,1), (B,1)\}$ are now found by executing the following query on~$C$ (and, again, this is not by accident):
\begin{quote}\small
\begin{verbatim}
SELECT   Building, COUNT(DISTINCT Emp) AS CNT
FROM     C
GROUP BY Building
\end{verbatim}
\end{quote}
Since $C$ can be expressed in SQL, we can actually construct a single SQL query that computes the lower bounds in $\cqacount{q_{0}}{z}$.

In general, if $q(\vec{z},\vec{w})$ is a full self-join-free conjunctive query for which $\cqacount{q}{\vec{z}}$ can be computed as previously described, then we will say that the query obtained from~$q$ by existentially binding the variables in~$\vec{w}$ (i.e., by binding the variables that are not grouping variables) admits \emph{parsimonious counting}.
Thus, our example showed that $\exists x\exists y\, E(x,\constant{F},y)\land D(y,z)$ admits parsimonious counting.
A formal definition of parsimonious counting will be given later on (Definition~\ref{def:pc}). 
In this introduction, we content ourselves by saying that parsimonious counting, if possible, computes $\cqacount{q}{\vec{z}}$ by executing two first-order queries (one for lower bounds, and one for upper bounds), followed by simple counting steps. 

\begin{figure}\centering\small
\begin{tabular}{ll}
\begin{tabular}{l|lcl}
$E$ & $\underline{\attribute{Emp}}$ & $\attribute{Gender}$ & $\attribute{Dept}$\bigstrut\\\cline{2-4}
& Suzy  & F & HR\\\cdashline{2-4}
& Anny  & F & HR\\
& Anny  & F & IT\\\cdashline{2-4}
& Dolores & F & IT\\\cdashline{2-4}
& Lucy & F & MIS
\end{tabular}
&
\begin{tabular}{l|lc}
$D$ & $\underline{\attribute{Dept}}$ & $\attribute{Building}$\bigstrut\\\cline{2-3}
& HR  & A\\\cdashline{2-3}
& IT  & A\\
& IT  & B\\\cdashline{2-3}
& MIS & B
\end{tabular}
\end{tabular}
\caption{Example database. Primary keys are underlined.}\label{fig:exampledatabase}
\end{figure}

The main contribution of our paper can now be described.
In his doctoral dissertation~\cite{FuxmanThesis}, Fuxman defined a class of self-join-free conjunctive queries, called $\cforest$, and showed the following.

\begin{theorem}[\cite{FuxmanThesis}]\label{the:fuxman}
Every query in $\cforest$ admits parsimonious counting.
\end{theorem}

The class $\cforest$ has been used in several studies on consistent query answering, because of its good computational properties.
It was an open question whether $\cforest$ contains \emph{all} self-join-free conjunctive queries that admit parsimonious counting.
We will answer this question negatively in Section~\ref{sec:cforest}. 
More fundamentally, we introduce a new syntactic class, called $\cnewclass$, which includes $\cforest$ and contains \emph{all} (and only) self-join-free conjunctive queries that admit parsimonious counting.
That is, we prove the following theorem.

\begin{theorem}[Main theorem]\label{the:main}
For every self-join-free conjunctive query $q$, it holds that $q$ admits parsimonious counting if and only if $q$ is in $\cnewclass$.
\end{theorem}

Moreover, a new and simpler proof for Theorem~\ref{the:fuxman} will follow in Section~\ref{sec:cforest}.

The remainder of this paper is organized as follows.
Section~\ref{sec:related} discusses related work.
Section~\ref{sec:preliminaries} introduces preliminary constructs and notations.
Section~\ref{sec:pc} introduces the semantic notion of parsimonious counting.
Section~\ref{sec:newclass} introduces our new syntactic class of queries, called $\cnewclass$, which restricts self-join-free conjunctive queries. 
Section~\ref{sec:soundness} shows that every query in $\cnewclass$ admits parsimonious counting, and Section~\ref{sec:completeness} shows that  $\cnewclass$ contains every self-join-free conjunctive query that admits parsimonious counting.
Section~\ref{sec:cforest} shows that $\cforest$ is strictly included in $\cnewclass$, and provides a new proof for Theorem~\ref{the:fuxman}.
Section~\ref{sec:conclusion} concludes the paper.
Several helping lemmas and proofs are in the appendix.

\section{Related Work}\label{sec:related}

Consistent query answering (CQA) started by a seminal paper in 1999 co-authored by Arenas, Bertossi, and Chomicki~\cite{DBLP:conf/pods/ArenasBC99}, who introduced the notions of repair and consistent answer. 
Two years later, the same authors introduced the \emph{range semantics} (with lower and upper bounds) for queries with aggregation~\cite{DBLP:conf/icdt/ArenasBC01,DBLP:journals/tcs/ArenasBCHRS03}\cite[Chapter~5]{DBLP:series/synthesis/2011Bertossi}, which has been commonly adopted ever since.
In particular, it was adopted in the PhD thesis~\cite{FuxmanThesis} of Fuxman, who provided Theorem~\ref{the:fuxman} (albeit using different terminology) and its proof, and used this result in the implementation of the ConQuer system~\cite{DBLP:conf/sigmod/FuxmanFM05}.
ConQuer aims at computations in first-order logic with counting (coined ``parsimonious counting'' in the current paper), which can be encoded in SQL. 
This is different from AggCAvSAT~\cite{DBLP:conf/icde/DixitK22}, a recent system by Dixit and Kolaitis, which uses powerful SAT solvers for computing range semantics, and thus can solve queries that are beyond the computational power of ConQuer.
Aggregation queries were also studied in the context of CQA in~\cite{DBLP:journals/is/BertossiBFL08}.


Consistent query answering for self-join-free conjunctive queries~$q$ and primary keys has been intensively studied. 
Its decision variant, which was coined $\cqa{q}$ in 2010~\cite{DBLP:conf/pods/Wijsen10}, asks whether a Boolean query~$q$ is true in every repair of a given database instance.
A systematic study of its complexity for self-join-free conjunctive queries had started already in~2005~\cite{DBLP:conf/icdt/FuxmanM05}, and was eventually solved in two journal articles by Koutris and Wijsen~\cite{DBLP:journals/tods/KoutrisW17,DBLP:journals/mst/KoutrisW21}, as follows: for every self-join-free Boolean conjunctive query~$q$, $\cqa{q}$ is either in $\FO$, $\L$-complete, or $\coNP$-complete, and it is decidable, given $q$, which case applies.
This complexity classification extends to non-Boolean queries by treating free variables as constants.
Other extensions beyond this trichotomy deal with foreign keys~\cite{DBLP:conf/pods/HannulaW22}, more than one key per relation~\cite{DBLP:conf/pods/KoutrisW20}, negated atoms~\cite{DBLP:conf/pods/KoutrisW18}, or restricted self-joins~\cite{DBLP:conf/pods/KoutrisOW21}.
For unions of conjunctive queries~$q$, Fontaine~\cite{DBLP:journals/tocl/Fontaine15} established interesting relationships between $\cqa{q}$  and Bulatov's dichotomy theorem
for conservative CSP~\cite{DBLP:journals/tocl/Bulatov11}.

The counting variant of $\cqa{q}$, denoted
$\sharp\cqa{q}$, asks to count the number of repairs that satisfy some Boolean query~$q$.
This counting problem is fundamentally different from the range semantics in the current paper.
For self-join-free conjunctive queries, $\sharp\cqa{q}$ exhibits a dichotomy between  $\FP$ and $\sharp\P$-complete under polynomial-time Turing reductions~\cite{DBLP:journals/jcss/MaslowskiW13}. This dichotomy has been shown to extend to queries with self-joins if primary keys are singletons~\cite{DBLP:conf/icdt/MaslowskiW14}, and to functional dependencies~\cite{DBLP:conf/pods/CalauttiLPS22a}.
Calautti, Console, and Pieris present in~\cite{DBLP:conf/pods/CalauttiCP19} a complexity analysis of these counting problems under weaker reductions, in particular, under many-one logspace reductions.
The same authors have conducted an experimental evaluation of randomized approximation schemes for approximating the percentage of repairs that satisfy a given query~\cite{DBLP:conf/pods/CalauttiCP21}.
Other approaches to making CQA more meaningful and/or tractable include operational repairs~\cite{DBLP:conf/pods/CalauttiLP18,DBLP:conf/pods/CalauttiLPS22} and preferred repairs~\cite{DBLP:journals/tcs/KimelfeldLP20,DBLP:journals/amai/StaworkoCM12}.

Recent overviews of two decades of theoretical research in CQA are \cite{DBLP:conf/pods/Bertossi19,DBLP:journals/sigmod/Wijsen19}.
It is worthwhile to note that theoretical research in $\cqa{q}$ has stimulated implementations and experiments in prototype systems~\cite{DBLP:conf/vldb/FuxmanFM05,DBLP:conf/sigmod/FuxmanFM05,DBLP:journals/pvldb/KolaitisPT13,DBLP:conf/sat/DixitK19,DBLP:conf/cikm/KhalfiouiJLSW20,DBLP:journals/corr/abs-2208-12339}.

\section{Preliminaries}\label{sec:preliminaries}

We assume that every relation name~$R$ is associated with an \emph{arity}, which is a positive integer.
We assume that all \emph{primary-key positions} precede all \emph{non-primary-key positions}.
We say that $R$ has \emph{signature} $\signature{n}{k}$ if $R$ has arity~$n$ and primary-key positions $1,\dots,k$.

If $R$ has signature $\signature{n}{k}$ and $s_{1},\dots,s_{n}$ are variables or constants, then $R(s_{1},\dots,s_{n})$ is an \emph{$R$-atom} (or simply \emph{atom}), which will often be denoted as $R(\underline{s_{1},\dots,s_{k}},s_{k+1},\dots,s_{n})$ to distinguish between primary-key and non-primary-key positions.
Two atoms $R_{1}(\underline{\vec{s}_{1}},\vec{t}_{1})$ and $R_{2}(\underline{\vec{s}_{2}},\vec{t}_{2})$ are said to be \emph{key-equal} if $R_{1}=R_{2}$ and $\vec{s}_{1}=\vec{s}_{2}$.
A \emph{fact} is an atom in which no variable occurs.
A \emph{database instance} (or simply \emph{database}) is a finite set of facts.
A database instance~$\db$ is \emph{consistent} if it does not contain two distinct key-equal facts.
A \emph{repair} of~$\db$ is a $\subseteq$-maximal consistent subset of~$\db$.

If $\vec{s}$ is a tuple of variables or constants, then $\arity{\vec{s}}$ denotes the arity of $\vec{s}$, and $\vars{\vec{s}}$ denotes the set of variables occurring in~$\vec{s}$.
By an abuse of notation, if we use a tuple $\vec{z}$ of variables at places where a set of variables is expected, we mean $\vars{\vec{z}}$. 
For an atom $F=R(\underline{\vec{s}},\vec{t})$, we define $\key{F}\defeq\vars{\vec{s}}$, $\notkey{F}\defeq\vars{\vec{t}}\setminus\vars{\vec{s}}$, and $\vars{F}\defeq\vars{\vec{s}}\cup\vars{\vec{t}}$.
For example, if $F=R(\underline{c,x,x,y},y,z,c)$, then $\key{F}=\{x,y\}$ and $\notkey{F}=\{z\}$, where $c$ is a constant.

\myheading{Conjunctive Queries.}
A conjunctive query~$q$ is a first-order formula of the form:
\begin{equation}\label{eq:cq}
\exists \vec{w}\formula{R_{1}(\underline{\vec{x}_{1}}, \vec{y}_{1}) \land \dotsm \land R_{n}(\underline{\vec{x}_{n}}, \vec{y}_{n})},
\end{equation}
where the variables of $\vec{w}$ are \emph{bound}, and the other variables are \emph{free}. 
Such a query is also denoted by $q(\vec{z})$ with $\vec{z}$ a tuple composed of the free variables. 
We write $\vars{q}$ for the set of variables that occur in $q$, and can assume $\vars{q}=\vars{\vec{w}}\cup\vars{\vec{z}}$ without loss of generality.
We say that $q$ is \emph{full} if all variables of $\vars{q}$ are free.
We say that $q$ is \emph{self-join-free} if $i\neq j$ implies $R_{i}\neq R_{j}$.
The quantifier-free part $R_{1}(\underline{\vec{x}_{1}}, \vec{y}_{1}) \land \dotsm \land R_{n}(\underline{\vec{x}_{n}}, \vec{y}_{n})$ of $q$ is denoted $\body{q}$.
By slightly overloading notation, we also use $\body{q}$ for the set $\{R_{1}(\underline{\vec{x}_{1}},\vec{y}_{1}), \ldots, R_{n}(\underline{\vec{x}_{n}},\vec{y}_{n})\}$. 
We write $\free{q}$ for the set of free variables in~$q$.

If a self-join-free conjunctive query~$q$ is understood, and we use a relation name~$R$ at places where an atom is expected, then we mean the unique $R$-atom of~$q$.
If $\vec{c}$ is a tuple of constants of arity $\arity{\vec{z}}$ and $\db$ a database instance,
then $\db\models q(\vec{c})$ denotes that $q(\vec{c})$ is true in $\db$ using standard first-order semantics. 
If $\db\models q(\vec{c})$, we also write $\vec{c}\in q(\db)$, and we say that $\vec{c}$ is an \emph{answer} to $q$ on $\db$.

We now introduce operators for turning bound variables into free variables, or vice versa, and for instantiating free variables.
\begin{description}
\item[Making bound variables free.]
Let $q$ be a conjunctive query with $\free{q}=\vec{z}$. 
Let $\vec{x}$ be a tuple of (not necessarily all) bound variables in $q$ (hence $\vec{x}\cap\vec{z}=\emptyset$).
We write $\makefree{q}{\vec{x}}$ for the conjunctive query $q'$ such that $\free{q'}=\vec{z}\cup\vec{x}$ and $\body{q'}=\body{q}$. 
Informally, $\makefree{q}{\vec{x}}$ is obtained from~$q$ by omitting the quantification $\exists\vec{x}$.
For example, if $q(z)=\exists x\exists y\, R(x,y)\land R(y,z)$,
then $\makefree{q}{x}=\exists y\, R(x,y)\land R(y,z)$.
\item[Binding free variables.]
Let $q$ be a conjunctive query, and $\vec{x}$ a tuple of (not necessarily all) free variables of~$q$.
Then $\makebound{q}{\vec{x}}$ denotes the query with the same body as $q$,
but whose set of free variables is $\free{q}\setminus\vec{x}$.
\item[Instantiating free variables.]
Let $q$ be a conjunctive query, and $\vec{z}$ a tuple of distinct free variables of~$q$. 
Let $\vec{c}$ be a tuple of constants of arity~$\arity{\vec{z}}$.
Then $\substitute{q}{\vec{z}}{\vec{c}}$ is the query obtained from~$q$ by replacing, for every $i\in\{1,2,\ldots,\arity{\vec{z}}\}$, each occurrence of the $i$th variable in~$\vec{z}$ by the $i$th constant in~$\vec{c}$.
\end{description}

\myheading{Consistent Query Answering.}
Let $q(\vec{z})$ be a conjunctive query.
We write $\db\cqamodels q(\vec{c})$ if for every repair $\rep$ of $\db$, we have $\rep\models q(\vec{c})$.
If $\db\cqamodels q(\vec{c})$, we also say that $\vec{c}$ is a \emph{consistent answer} to $q$ on~$\db$.
A \emph{consistent first-order rewriting} of $q(\vec{z})$ is a first-order formula $\varphi(\vec{z})$ such that for every database instance $\db$ and every tuple $\vec{c}$ of constants of arity~$\arity{\vec{z}}$,
we have $\db\cqamodels q(\vec{c})$ if and only if $\db\models\varphi(\vec{c})$.
Note incidentally that the set of integrity constraints is always implicitly understood to be the primary keys associated with the relation names that occur in the query.
 




\myheading{Query Graph.}
The \emph{query graph} of a conjunctive query $q(\vec{z})$ is an undirected graph whose vertices are the bound variables of $q$.
There is an edge between $x$ and $y$ if $x\neq y$ and $x,y$ occur together in some atom of $\body{q}$.

\myheading{Attack Graph.}
The following is a straightforward extension of attack graphs~\cite{DBLP:journals/tods/KoutrisW17} to deal with free variables.

Let $q(\vec{z})$ be a self-join-free conjunctive query.
If $S$ is a subset of $\body{q}$, then $q\setminus S$ denotes the query obtained from $q$ by removing from~$q$ all atoms in~$S$.
Every variable of $q\setminus S$ that is free in $q$ remains free in~$q\setminus S$;
and every variable of $q\setminus S$ that is bound in $q$ remains bound in  $q\setminus S$.

We define $\fdset{q}$ as the set of functional dependencies that contains $\fd{\emptyset}{\free{q}}$ and contains,
for every atom $F$ in~$q$, the functional dependency $\fd{\key{F}}{\vars{F}}$.
Note that since $\fdset{q}$ contains $\fd{\emptyset}{\free{q}}$, we have that $\fdset{q}\models\fd{\emptyset}{y}$ if and only if $\fdset{q}\models\fd{\free{q}}{y}$, for each $y\in\vars{q}$.
If $F$ is an atom of $q$, then $\keycl{F}{q}$ is the set that contains every variable $y\in\vars{q}$ such that either $y\in\free{q}$ or $\fdset{q\setminus\{F\}}\models\fd{\key{F}}{y}$ (or both).

It is known~\cite{DBLP:journals/tods/KoutrisW17} that in the study of consistent query answering for self-join-free conjunctive queries, free variables can often be treated as constants. 
The addition of functional dependencies $\fd{\emptyset}{\free{q}}$ has the same effect as treating variables in $\free{q}$ as constants.
In the following example, we omit curly braces and commas when denoting sets of variables.
For example, $\{z_{1},z_{2}\}$ is denoted $z_{1}z_{2}$.

\begin{example}\label{ex:fdset}
Let 
$q=\exists u\exists v\exists x\exists y\, R(\underline{u},x)\land S(\underline{x,z_{1}},y)\land T(\underline{y},v,z_{2})\land U(\underline{y},u)$.
We have $\free{q}=z_{1}z_{2}$.
Then, $q\setminus\{T\}$\footnote{Recall that we use $T$ as a shorthand for the $T$-atom of~$q$.} is the query 
$\exists u\exists v\exists x\exists y\, R(\underline{u},x)\land S(\underline{x,z_{1}},y)\land U(\underline{y},u)$, 
whose only free variable is~$z_{1}$.
Note incidentally that since $v$ does not occur in the latter query, the quantification $\exists v$ can be dropped.
We have $\fdset{q\setminus\{T\}}=\{\fd{\emptyset}{z_{1}}, \fd{u}{x}, \fd{xz_{1}}{y}, \fd{y}{u}\}$.
Note that $\fd{\emptyset}{z_{1}}$ belongs to the latter set because $z_{1}$ is free in~$q\setminus\{T\}$.
The closure of $\key{T}$ with respect to $\fdset{q\setminus\{T\}}$ is $uxyz_{1}$.
Finally, we obtain $\keycl{T}{q}=uxyz_{1}z_{2}$.
Note that the latter set contains the variable $z_{2}$ that is free in~$q$.
\qed
\end{example}

We say that an atom $F$ of $q$ \emph{attacks} a variable $x$ occurring in $q$, denoted $\attacks{F}{x}{q}$, if there exists a sequence 
\begin{equation}\label{eq:witness}
\tuple{x_{1}, x_{2}, \ldots, x_{n}}
\end{equation}
of bound variables of~$q$ ($n\geq 1$) such that:
\begin{enumerate}
\item
if two variables are adjacent in the sequence, then they occur together in some atom of $q$;
\item
$x_{1}\in\notkey{F}$ and $x_{n}=x$; and
\item
for every $\ell\in\{1,\ldots,n\}$, $x_{\ell}\not\in\keycl{F}{q}$.
\end{enumerate}
The sequence~\eqref{eq:witness} will be called a \emph{witness} that $\attacks{F}{x}{q}$.
We say that an atom $F$ of $q$ attacks another atom~$G$ of~$q$, denoted $\attacks{F}{G}{q}$, if $F\neq G$ and $F$ attacks some variable of $\vars{G}$.
It is now easily verified that if~$F$ attacks $G$, then $F$ also attacks a variable in $\key{G}$.
A variable or atom that is not attacked, is called \emph{unattacked} (where~$q$ is understood from the context).
The \emph{attack graph} of $q$ is a directed graph whose vertices are the atoms of $q$;
there is a directed edge from $F$ to $G$ if $\attacks{F}{G}{q}$.
A directed edge in the attack graph is called an \emph{attack}.

Koutris and Wijsen~\cite{DBLP:journals/tods/KoutrisW17} showed the following.
\begin{theorem}[\cite{DBLP:journals/tods/KoutrisW17}]\label{the:koutriswijsen}
A self-join-free conjunctive query $q(\vec{z})$ has a consistent first-order rewriting if and only if its attack graph is acyclic.
\end{theorem}

An attack from $F$ to $G$ is \emph{weak} if $\fdset{q}$ logically implies $\fd{\key{F}}{\key{G}}$; otherwise it is strong.
By a \emph{component} of an attack graph, we always mean a maximal weakly connected component.


Let $q$ be a self-join-free conjunctive query.
Whenever the relationship $\fdset{q}\models\fd{Z}{w}$ holds true, then there exists a sequential proof of it, as defined next.

\myheading{Sequential Proof.}
Let $q(\vec{z})$ be a self-join-free conjunctive query, and $Z\subseteq\vars{q}$.
Let $\tuple{F_{1}, F_{2}, \ldots, F_{n}}$ be a (possibly empty) sequence of atoms in $\body{q}$ such that for every $i\in \{1,\ldots,n\}$, 
$\key{F_i}\subseteq\free{q}\cup Z\cup\formula{\bigcup_{j=1}^{i-1} \vars{F_j}}$.
Such a sequence is called a \emph{sequential proof} of $\fdset{q}\models\fd{Z}{w}$, for every $w\in\free{q}\cup Z\cup\formula{\bigcup_{j=1}^{n} \vars{F_j}}$.
A sequential proof of $\fdset{q}\models\fd{Z}{w}$ is called \emph{minimal} if $\tuple{F_{1},\ldots,F_{n-1}}$ is not a sequential proof of $\fdset{q}\models\fd{Z}{w}$.


%


\section{Parsimonious Counting}\label{sec:pc}

Consider a conjunctive query $q(\vec{z})=\exists\vec{w}\; B$, with $B$ a quantifier-free conjunction of atoms (called the \emph{body}). 
We introduce a query that takes a database instance $\db$ as input and returns, for every tuple $\vec{c}\in q(\db)$, the number of valuations for $\vec{w}$ that make the query true. 

\begin{definition}[$\mycount{q}{\vec{z}}$]\label{def:cntsimplified}
Let $q(\vec{z},\vec{w})$ be a full conjunctive query, in which notation it is understood that $\vec{z}$ and $\vec{w}$ are disjoint, duplicate-free tuples of variables. 
$\mycount{q}{\vec{z}}$ is the query that takes as input a database instance $\db$ and returns every tuple $(\vec{c},i)$ for which the following hold:
\begin{enumerate}
\item 
$\vec{c}$ a tuple of constants of arity $\arity{\vec{z}}$; and 
\item
$i$ is a positive integer such that $i$ is the number of distinct tuples $\vec{d}$, of arity~$\arity{\vec{w}}$, satisfying $\db\models q(\vec{c},\vec{d})$. 
\end{enumerate}
A maximal set of answers to $q(\db)$ that agree on $\vec{z}$ will also be called a \emph{$\vec{z}$-group} (where $q$ and $\db$ are implicitly understood).
Thus, $\mycount{q}{\vec{z}}$ counts the number of tuples in each $\vec{z}$-group.
\end{definition}

The following definition introduces range consistent query answers as introduced in~\cite{DBLP:conf/icdt/ArenasBC01}.

\begin{definition}[$\cqacount{q}{\vec{z}}$]\label{def:cqacnt}
Let $q(\vec{z},\vec{w})$ be a full conjunctive query, in which notation it is understood that $\vec{z}$ and $\vec{w}$ are disjoint, duplicate-free tuples of variables. 
$\cqacount{q}{\vec{z}}$ is the query that takes as input a database instance $\db$ and returns every tuple $(\vec{c}, [m,n])$ for which the following hold:
\begin{enumerate}
    \item\label{it:cqacntone} for every repair $\rep$ of $\db$, there exists $\vec{d}$ such that $\rep\models q(\vec{c},\vec{d})$; 
    \item\label{it:cqacnttwo} there is a repair of $\db$ on which $\mycount{q}{\vec{z}}$ returns $(\vec{c},m)$; 
    \item\label{it:cqacntthree} there is a repair of $\db$ on which $\mycount{q}{\vec{z}}$ returns $(\vec{c},n)$; and
    \item\label{it:cqacntfour} if $\mycount{q}{\vec{z}}$ returns $(\vec{c},i)$ on some repair of $\db$, then $m\leq i\leq n$.
\end{enumerate}
Note that it follows that $m\geq 1$.
If  $(\vec{c}, [m,n])$ is an answer to $\cqacount{q}{\vec{z}}$ on $\db$,
then we will say that it is a \emph{range-consistent answer}.
\end{definition}

The following proposition states that computing $\cqacount{q(\vec{z},\vec{w})}{\vec{z}}$ can be $\NP$-hard, even if the query $\makebound{q}{\vec{w}}$ has a consistent first-order rewriting.

\begin{proposition}\label{pro:threedm}
There exists a self-join-free conjunctive query $q(\vec{z})$ that has a consistent first-order rewriting such that $\cqacount{\body{q}}{\vec{z}}$
is $\NP$-hard to compute.
\end{proposition}
\begin{proof}
The following problem is $\NP$-complete~\cite{DBLP:books/fm/GareyJ79}.
\begin{description}
\item[\problem{3-DIMENSIONAL MATCHING} (\problem{3DM})]
\item[INSTANCE:] A set $M\subseteq A_{1}\times A_{2}\times A_{3}$, where $A_{1}$, $A_{2}$, $A_{3}$ are disjoint sets 
having the same number $n$ of elements. 
\item[QUESTION:] Does $M$ contain a matching, that is, a subset $M'\subseteq M$ such 
that $\card{M'}=n$ and no two elements of $M'$  agree in any coordinate? 
\end{description}
Consider the query 
$$
q(z)=\exists x_{1}\exists x_{2}\exists x_{3}\exists y\, Z(\underline{z})\land\bigwedge_{i=1}^{3}\formula{R_{i}(\underline{x_{i}},y)\land S_{i}(\underline{x_{i}},y)}.
$$
The edge-set of $q$'s attack graph is empty.
Therefore, $q$'s attack graph is acyclic.
By Theorem~\ref{the:koutriswijsen}, $q(z)$ has a consistent first-order rewriting.
Let $M\subseteq A_{1}\times A_{2}\times A_{3}$ be an instance of \problem{3DM}.
Let $\db_{M}$ be the database instance that contains $Z(\underline{c})$ and includes,
for every $a_{1}a_{2}a_{3}$ in $M$, $\bigcup_{i=1}^{3}\{R_{i}(\underline{a_{i}},a_{1}a_{2}a_{3}), S_{i}(\underline{a_{i}},a_{1}a_{2}a_{3})\}$.
Moreover, $\db_{M}$ includes $D\defeq\bigcup_{i=1}^{3}\{R_{i}(\underline{\bot_{i}},\top), S_{i}(\underline{\bot_{i}},\top)\}$,
where $\bot_{1}, \bot_{2}, \bot_{3}, \top$ are fresh constants not in $A_{1}\cup A_{2}\cup A_{3}$.
Clearly, $\db_{M}$ is first-order computable from~$M$.

We show the following: $M$ has a matching
if and only if
for some $\ell$, $\cqacount{\body{q}}{z}$ returns $(c,[\ell,n+1])$ on $\db_{M}$.
Before delving into the proof, we provide an example.
\begin{example}
Let $A_{1}=\{a,b\}$, $A_{2}=\{d,e\}$, $A_{3}=\{f,g\}$, and $M=\{adf, aeg, beg\}$.
Thus, $n=\card{A_{1}}=2$.
Then, we construct relations as follows.
\begin{small}
$$
\setlength{\arraycolsep}{0.5\arraycolsep}
\begin{array}{cccc}
\db_{M}:
\begin{array}{c|*{1}{c}c}
Z & \underline{z}\bigstrut\\\cline{2-2}
  & c & \ast
\end{array}
&
\begin{array}{c|*{2}{c}c}
R_{1}=S_{1} & \underline{x_{1}} & y\bigstrut\\\cline{2-3}
  & a & adf & \ast\\
  & a & aeg\\\cdashline{2-3}
  & b & beg & \ast\\\cdashline{2-3}
  & \bot_{1} & \top & \ast
\end{array}
&
\begin{array}{c|*{2}{c}c}
R_{2}=S_{2} & \underline{x_{2}} & y\bigstrut\\\cline{2-3}
  & d & adf & \ast\\\cdashline{2-3}
  & e & aeg\\
  & e & beg & \ast\\\cdashline{2-3}
  & \bot_{2} & \top & \ast
\end{array}
&
\begin{array}{c|*{2}{c}c}
R_{3}=S_{3} & \underline{x_{3}} & y\bigstrut\\\cline{2-3}
 & f & adf & \ast\\\cdashline{2-3}
 & g & aeg\\
 & g & beg & \ast\\\cdashline{2-3}
 & \bot_{3} & \top & \ast
\end{array}
\end{array}
$$
\end{small}
Let $q^{*}(z,x_{1},x_{2},x_{3},y)$ be the full query whose set of atoms is $\body{q}$.
We have $\db_{M}\cqamodels q(c)$ because $(c,\bot_{1},\bot_{2},\bot_{3},\top)$ is in the answer to $q^{*}$ on every repair.
It can be verified that $\cqacount{q^{*}}{z}$ returns $(c,[1,3])$ on $\db_{M}$.
The tuples marked with $\ast$ form a repair $\rep$ on which $\mycount{q^{*}}{z}$ returns $(c,3)$.
In particular,  $q^{*}(\rep)=\{(c,\bot_{1},\bot_{2},\bot_{3},\top)$, $(c,a,d,f,adf)$, $(c,b,e,g,beg)\}$.
This answer corresponds to the 3-dimensional matching $\{adf, beg\}$.

Note that if we remove from $\db_{M}$ one of the facts $R_{i}(\underline{\bot_{i}},\top)$ or $S_{i}(\underline{\bot_{i}},\top)$, for some $1\leq i\leq 3$,
then $c$ is no longer a consistent answer to $q(z)$, and, according to Definition~\ref{def:cqacnt}, the answer to $\cqacount{\body{q}}{z}$ would be empty.
\end{example}
\framebox{$\implies$}
Assume $M'$ is a matching of~$M$.
Let $\rep$ be a database instance that includes~$D$ and includes, for every $a_{1}a_{2}a_{3}\in M'$, the set $\bigcup_{i=1}^{3}\{R_{i}(\underline{a_{i}},a_{1}a_{2}a_{3})$, $S_{i}(\underline{a_{i}},a_{1}a_{2}a_{3})\}$.
Since no two elements of $M'$ agree on any coordinate, $\rep$ is consistent.
Moreover, since $n=\card{M'}=\card{A_{1}}=\card{A_{2}}=\card{A_{3}}$, $\rep$ contains a tuple of every block of $\db_{M}$.
Therefore, $\rep$ is a repair of $\db_{M}$.
Clearly, $\mycount{q}{z}$ returns $(c,n+1)$ on $\rep$.
It is also obvious to see that there is no repair of $\db_{M}$ on which $\mycount{\body{q}}{z}$ returns  $(c,k)$ with $k>n+1$.
Moreover, $\db_{M}\cqamodels q(c)$ because $D$ is included in every repair of~$\db_{M}$.
Therefore, for some~$\ell$, $\cqacount{\body{q}}{z}$ returns $(c,[\ell,n+1])$ on $\db_{M}$.

\framebox{$\impliedby$}
Assume that for some $\ell$, $\cqacount{\body{q}}{z}$ returns $(c,[\ell,n+1])$ on $\db_{M}$.
Then, there is a repair $\rep$ on which $\mycount{\body{q}}{z}$ returns~$n+1$.
Let $q^{*}(z,x_{1},x_{2},x_{3},y)$ be the full query whose set of atoms is $\body{q}$.
Let $(c,a_{1},a_{2},a_{3},b_{1}b_{2}b_{3})$, with $b_{1}b_{2}b_{3}\neq\top$, be a tuple in $q^{*}(\rep)$. 
For each $i\in\{1,2,3\}$, since $q^{*}$ contains the atom $R_{i}(\underline{x_{i}},y)$, $\rep$ contains~$R_{i}(\underline{a_{i}},b_{1}b_{2}b_{3})$, and hence $a_{i}=b_{i}$ by our construction.
It follows $b_{1}b_{2}b_{3}=a_{1}a_{2}a_{3}$.

Let $(c,a_{1},a_{2},a_{3},a_{1}a_{2}a_{3})$ and $(c,b_{1},b_{2},b_{3},b_{1}b_{2}b_{3})$ be two tuples in $q^{*}(\rep)$.
Since $q^{*}$ contains $R_{i}(\underline{x_{i}},y)$,
$\rep$ contains $R_{i}(\underline{a_{i}},a_{1}a_{2}a_{3})$ and $R_{i}(\underline{b_{i}},b_{1}b_{2}b_{3})$.
If $a_{i}=b_{i}$, then, since $\rep$ is consistent, we have $a_{1}a_{2}a_{3}=b_{1}b_{2}b_{3}$.
Consequently, no two distinct answers in  $q^{*}(\rep)$ agree on any coordinate among $x_{1}$, $x_{2}$, and $x_{3}$.
Since  $\mycount{\body{q}}{z}$ returns~$n+1$ on $\rep$, it follows that the set $\{a_{1}a_{2}a_{3}\mid\rep\setminus D\models q^{*}(c,a_{1},a_{2},a_{3},a_{1}a_{2}a_{3})\}$ is a matching of $M'$ of size~$n$.
This concludes the proof.
\end{proof}

Note that the foregoing proof can be easily adapted from \problem{3-DIMENSIONAL MATCHING} to \problem{2-DIMENSIONAL MATCHING}.
That is, the query $q(z)=\exists x_{1}\exists x_{2}\exists y\, Z(\underline{z})\land\bigwedge_{i=1}^{2}\formula{R_{i}(\underline{x_{i}},y)\land S_{i}(\underline{x_{i}},y)}$ has a consistent first-order rewriting, but computing $\cqacount{\body{q}}{z}$ is as hard as \problem{$2$-DIMENSIONAL MATCHING}.

The following definition introduces the semantic notion of \emph{parsimonious counting}, which was illustrated by the running example in Section~\ref{sec:introduction}.
Informally, for a query $q(\vec{z})$ that admits parsimonious counting, it will be the case that on every database instance $\db$,
the answers to $\cqacount{\body{q}}{\vec{z}}$ can be computed by a first-order query followed by a simple counting step.

\begin{definition}[Parsimonious counting]\label{def:pc}
Let $q$ be a conjunctive query with $\free{q}=\vec{z}$.\footnote{We will commonly write $q(\vec{z})$ to make explicit that $\free{q}=\vec{z}$.}
Let~$\vec{x}$ be a (possibly empty) sequence of distinct bound variables of $q(\vec{z})$.
We say that $q$ \emph{admits parsimonious counting on~$\vec{x}$} if the following hold (let $q'(\vec{z},\vec{x})=\makefree{q}{\vec{x}}$):
\begin{enumerate}[label=(\Alph*)]
    \item\label{it:pcone} $q(\vec{z})$ has a consistent first-order rewriting;
    \item\label{it:pctwo} $q'(\vec{z},\vec{x})$ has a consistent first-order rewriting (call it $\varphi(\vec{z},\vec{x})$); and
    \item\label{it:pcthree} for every database instance $\db$, the following conditions~\ref{it:pcthreeone} and~\ref{it:pcthreetwo} are equivalent:
    \begin{enumerate}[label=(\alph*)]
    \item\label{it:pcthreeone}
    $(\vec{c},[m,n])$ is an answer to $\cqacount{\body{q}}{\vec{z}}$ on $\db$;
    \item\label{it:pcthreetwo}
    $m\geq 1$ and both the following hold:
    \begin{enumerate}[label=(\roman*)]
    \item\label{it:pcthreetwoone} $m$ is the number of distinct tuples $\vec{d}$, of arity $\vec{x}$, such that $\db\models \varphi(\vec{c},\vec{d})$; and
    \item\label{it:pcthreetwotwo} $n$ is the number of distinct tuples $\vec{d}$ such that $\db\models q'(\vec{c},\vec{d})$.
\end{enumerate}    
\end{enumerate}
\end{enumerate}
We say that $q$ \emph{admits parsimonious counting} if it admits parsimonious counting on some sequence $\vec{x}$ of bound variables.
\end{definition}

Significantly, since Definition~\ref{def:pc} contains a condition that must hold for every database instance $\db$, it does not give us an efficient procedure for deciding whether a given self-join-free query $q(\vec{z})$ admits parsimonious counting. 

We now give some examples.
From the proof of Proposition~\ref{pro:threedm} and the paragraph after that proof, 
it follows that under standard complexity assumptions, for $k\geq 2$, 
$$q_{k}(z)\defeq\exists x_{1}\dotsm\exists x_{k}\exists y\, Z(\underline{z})\land\bigwedge_{i=1}^{k}\formula{R_{i}(\underline{x_{i}},y)\land S_{i}(\underline{x_{i}},y)}$$ does not admit parsimonious counting, even though $q_{k}(z)$ has a consistent first-order rewriting.
The following example shows a query~$q(z)$ that does not admit parsimonious counting, but for which $\cqacount{\body{q}}{z}$ can be computed in first-order logic with a counting step that is slightly more involved than what is allowed in parsimonious counting.

\begin{example}\label{ex:almostpc}
Let $q(z)=\exists x\exists y\, R(\underline{z},x)\land S(\underline{x,y})$ and $q^{*}(z,x,y)=R(\underline{z},x)\land S(\underline{x,y})$.
We first argue that $q(z)$ does not admit parsimonious counting.
Let $\db$ be the following database instance:
\begin{small}
$$
\begin{array}{cc}
\begin{array}{c|cc}
R    & \underline{z} & x\\\cline{2-3}
     & c_{1} & a\\\cdashline{2-3}
     & c_{2} & a\\
     & c_{2} & b
\end{array}
&
\begin{array}{c|cc}
S    & \underline{x} & \underline{y}\bigstrut\\\cline{2-3}
     & a & d\\\cdashline{2-3}
     & a & e\\\cdashline{2-3}
     & b & f
\end{array}
\end{array}
$$
\end{small}
It can be verified that on this database instance, $\cqacount{q^{*}}{z}$ must return $(c_{1},[2,2])$ and $(c_{2},[1,2])$.
We next show the answer to $q^{*}$ on $\db$:
\begin{small}
$$
\begin{array}{c|ccc}
q^{*}(\db)    & z & x & y\\\cline{2-4}
     & c_{1} & a & d\\
     & c_{1} & a & e\\
     & c_{2} & a & d\\
     & c_{2} & a & e\\
     & c_{2} & b & f
\end{array}
$$
\end{small}
The correct upper bound of~$2$ in $(c_{2},[1,2])$ could only be obtained by counting, within the $c_{2}$-group, the number of distinct $\tuple{x}$-values. 
However, such a counting would conclude an incorrect upper bound of~$1$ for the $c_{1}$-group.
It is now correct to conclude that $q(z)$ does not admit parsimonious counting.

The lower and upper bounds can be obtained from $q^{*}(\db)$ by a counting step that is only slightly more involved than what is allowed in parsimonious counting. 
First, construct the following relation where $\tilde{R}(\underline{c_{j}},v\mid n)$ means that $\mycount{q^{*}}{z}$ returns $(c_{j},n)$ on a repair that contains $R(\underline{c_{j}},v)$.
\begin{small}
$$
\begin{array}{c|cc|c}
\tilde{R}    & \underline{z} & x\\\cline{2-4}
     & c_{1} & a & 2\\\cdashline{2-4}
     & c_{2} & a & 2\\
     & c_{2} & b & 1
\end{array}
$$
\end{small}
These counts can be obtained from $q^{*}(\db)$ by counting the number of distinct $y$-values within each $zx$-group.
Next, the lower and upper bounds are obtained as the minimal and maximal counts within each $z$-group.

Note incidentally that for $q_{0}(z)\defeq\exists x\exists y\, R(\underline{z},x)\land T(\underline{z},x)\land S(\underline{x,y})$,
which is obtained from $q(z)$ by adding $T(\underline{z},x)$, we have that $q_{0}$ admits parsimonious counting.
The change occurs because if $\db\cqamodels q_{0}(c)$, then there exists a unique value $a$ such that
$\db\models\forall x\formula{R(\underline{c},x)\rightarrow x=a}$ and $\db\models\forall x\formula{T(\underline{c},x)\rightarrow x=a}$.
That is, the only blocks that can contribute to $\cqacount{\body{q_{0}}}{z}$ have cardinality~$1$.
This means that range semantics reduces to counting on a consistent database instance.
\qed
\end{example}

\section{The Class $\cnewclass$}\label{sec:newclass}

The notion of parsimonious counting is a semantic property defined for conjunctive queries.
A natural question is to syntactically characterize the class of conjunctive queries that admit parsimonious counting.
In this paper, we will answer this question under the restriction that queries are self-join-free.
This is the best we can currently hope for, because consistent query answering for primary keys and conjunctive queries with self-joins is a notorious open problem for which no tools are known (e.g., attack graphs are not helpful in the presence of self-joins). 
We now define our new syntactic class $\cnewclass$, which uses the following notion of \emph{frozen variable}. 

\begin{definition}[Frozen variable]\label{def:frozen}
Let $q(\vec{z})$ be a self-join-free conjunctive query.
We say that a bound variable~$y$ of $q(\vec{z})$ is \emph{frozen in~$q$} if there exists a sequential proof of $\fdset{q}\models\fd{\emptyset}{y}$ such that $\notattacks{F}{y}{q}$ for every atom $F$ that occurs in the sequential proof.
We write $\frozen{q}$ for the set of all bound variables of $\vars{q}$ that are frozen in~$q$.
A bound variable that is not frozen in~$q$ is called \emph{nonfrozen in $q$}.
\qed
\end{definition}

\begin{example}
Let $q(z)=\exists x\, R(\underline{z},x)\land S(\underline{z},x)$.
We have $\notattacks{R}{x}{q}$.
Therefore, $\tuple{R(\underline{z},x)}$ is a sequential proof of $\fdset{q}\models\fd{\emptyset}{x}$ that uses no atom attacking~$x$. Hence, $x$ is frozen.
Note here that $z$ is free, hence $\fdset{q}\models\fd{\emptyset}{z}$ by definition.
\qed
\end{example}

\begin{definition}[The class $\cnewclass$]\label{def:newclass}
We define $\cnewclass$ as the set of self-join-free conjunctive queries $q(\vec{z})$ satisfying the following conditions:
\begin{enumerate}[label=(\Roman*)]
    \item\label{it:noCyclicOrStrongAttacks} the attack graph of $q(\vec{z})$ is acyclic and contains no strong attacks; and
    \item\label{it:newclasstwo} 
    there is a tuple $\vec{x}$ of bound variables of $q(\vec{z})$ such that:
			\begin{enumerate}[label=(\arabic*)]
				\item\label{it:source} every component\footnote{Whenever we use the term component, we mean a maximal weakly connected component.} of $q(\vec{z})$'s attack graph contains an unattacked atom $R$ such that $\fdset{q}\models\fd{\vec{x}}{\key{R}}$; and
 				\item\label{it:separation} for every atom $R$ in $\body{q(\vec{z})}$, every (possibly empty) path in the query graph of $q(\vec{z})$ between a variable of $\notkey{R}$ and a variable of $\vec{x}$ uses a variable in $\key{R}\cup\frozen{q}$.
			\end{enumerate}
\end{enumerate}
We will say that such an $\vec{x}$ is an \emph{id-set} for $q(\vec{z})$.
We will say that an id-set $\vec{x}$ is \emph{minimal} if any sequence obtained from $\vec{x}$ by omitting one or more variables is no longer an id-set. 
\end{definition}

Informally, id-sets $\vec{x}$ will play the role of $\vec{x}$ in Definition~\ref{def:pc}: they identify the values that have to be counted within each $\vec{z}$-group to obtain range-consistent answers.

\begin{figure}
    \centering
    \begin{tabular}{c@{\hspace{7em}}c}
         \includegraphics[scale=0.8]{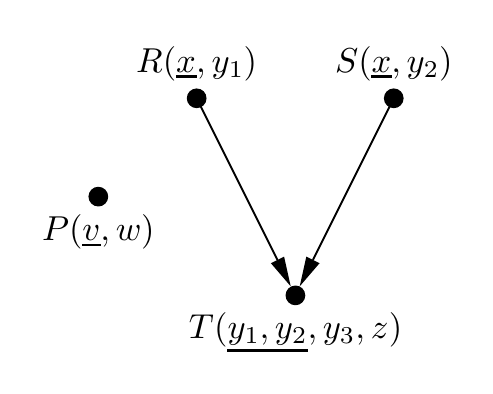}
         &
         \includegraphics[scale=0.8]{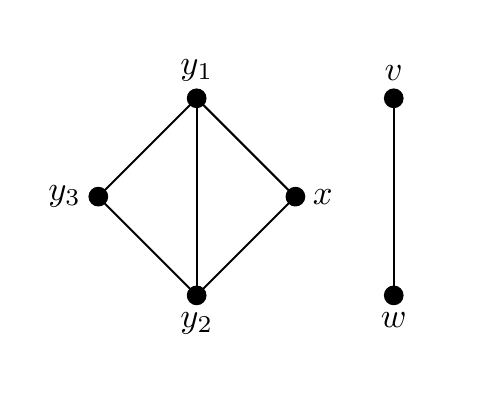}
    \end{tabular}
    
    \caption{Attack graph \emph{(left)} and query graph \emph{(right)} of $q(z) = \exists x\exists y_1\exists y_2\exists y_3\exists v\exists w\, R(\underline{x}, y_1) \land S(\underline{x}, y_2) \land T(\underline{y_1, y_2}, y_3, z) \land P(\underline{v}, w)$.}
    \label{fig:graphs}
\end{figure}

We now illustrate Definition~\ref{def:newclass} by some examples.
Then Proposition~\ref{pro:defineX} implies that every query $q(\vec{z})$ in $\cnewclass$ has a unique minimal id-set that can be easily constructed from $q(\vec{z})$'s attack graph.

\begin{example}
In the paragraph following the proof of Proposition~\ref{pro:threedm}, we introduced the query
$q(z)=\exists x_{1}\exists x_{2}\exists y\, Z(\underline{z})\land\bigwedge_{i=1}^{2}\formula{R_{i}(\underline{x_{i}},y)\land S_{i}(\underline{x_{i}},y)}$.
The edge-set of $q(z)$'s attack graph is empty.
No variable is frozen.
According to condition~\ref{it:source} in Definition~\ref{def:newclass}, every id-set (if any) must contain $x_{1}$.
However, no id-set can contain $x_{1}$, because for the atom $R_{2}(\underline{x_{2}},y)$, the edge $\{y,x_{1}\}$ in the query graph is a path between a variable of $\notkey{R_{2}}$ and $x_{1}$ that uses no variable of $\key{R_{2}}$. 
We conclude that $q(z)$ is not in $\cnewclass$.
\qed
\end{example}

\begin{example}
The query $q(z)=\exists x\exists y\exists v\, R(\underline{x},y)\land S(\underline{y},v)\land T(\underline{v},y)\land P_{1}(\underline{z},y)\land P_{2}(\underline{z},y)$ belongs to $\cnewclass$.
The attack graph of $q(z)$ has a single attack from $S$ to~$T$. 
The query graph of $q(z)$ has two undirected edges: $\{x,y\}$ and $\{y,v\}$.
The variable $y$ is frozen, because $\tuple{P_{1}(\underline{z},y)}$ is a sequential proof of $\fdset{q}\models\fd{\emptyset}{y}$ (note here that $z$ is free), and $\notattacks{P_{1}}{y}{q}$.

It can be verified that $\tuple{x}$ is an id-set.
Note that $\tuple{y,x}$ is a path in the query graph between $y\in\notkey{T}$ and $x$ that uses no variable of~$\key{T}=\{v\}$.
However, that path uses the frozen variable~$y$.
\qed
\end{example}

\begin{example}
Let $q(z) = \exists x\exists y_1\exists y_2\exists y_3\exists v\exists w\, R(\underline{x}, y_1) \land S(\underline{x}, y_2) \land T(\underline{y_1, y_2}, y_3, z) \land P(\underline{v}, w)$. 
The attack graph and the query graph of~$q(z)$ are shown in Fig.~\ref{fig:graphs}.
We now argue that $q(z)$ is in $\cnewclass$.
First, the attack graph of $q$ is acyclic and contains no strong attacks.
We next argue that $xv$ is an id-set for~$q$.
The attack graph of $q(z)$ has two components.
Condition~\ref{it:source} in Definition~\ref{def:newclass} is obviously satisfied for $\vec{x}=xv$ since $\fdset{q} \models \fd{xv}{v}$ and $\fdset{q} \models \fd{xv}{x}$.
It is easily verified that condition~\ref{it:separation} is also verified.
In particular, for the atom $T(\underline{y_1, y_2}, y_3, z)$, every path between $y_{3}$ and $x$ uses either $y_{1}$ or $y_{2}$. 
\qed
\end{example}

\begin{example}
Let $q(z) = \exists x\exists y\, R_{1}(\underline{x},y,z) \land R_{2}(\underline{x},y) \land S_{1}(\underline{y},x) \land S_{2}(\underline{y},x)$.
The attack graph of $q(z)$ contains no edges and, thus, is acyclic and has four components.
It can be verified that no variable is frozen.
We claim that $q(z)$ is not in $\cnewclass$, because it has no id-set.
Indeed, from condition~\ref{it:source} in Definition~\ref{def:newclass}, it follows that every id-set must contain either $x$ or $y$ (or both).
For the atom $S_{1}(\underline{y},x)$, the empty path is a path between a variable in $\notkey{S_{1}}$ to~$x$ that uses no variable in $\key{S_{1}}$.
It follows by condition~\ref{it:separation} that no id-set can contain~$x$.
From $R_{2}(\underline{x},y)$, by similar reasoning, we conclude that no id-set can contain~$y$.
It follows that $q(z)$ has no id-set.
\qed
\end{example}

\begin{proposition} \label{pro:defineX}
Let $q(\vec{z})$ be a query in $\cnewclass$, and let $\vec{x}$ be a minimal id-set for it. 
Let $N=\bigcup\{\notkey{R}\mid R\in q\}$.
Let $V$ be a $\subseteq$-minimal subset of $\vars{q}$ that includes, for every unattacked atom $R$ of $q$, every bound variable of $\key{R}\setminus N$.
Then, 
\begin{enumerate}[label=(\Alph*)]
    \item\label{it:defineXone} $V=\vars{\vec{x}}$; and
    \item\label{it:defineXtwo} whenever $R,S$ are unattacked atoms that are weakly connected in $q(\vec{z})$'s attack graph, $\key{R}\cap\vec{x}=\key{S}\cap\vec{x}$.
\end{enumerate}
\end{proposition}

\begin{example}
Let $q(z_{1},z_{2})=\exists x\exists y\, R(\underline{x},y,z_{1})\land S(\underline{x},y)\land T(\underline{y},z_{2})$.
The attack graph of~$q$ contains no edges, and hence has three weak components. 
If we construct $V$ as in the statement of Proposition~\ref{pro:defineX}, we first compute $N=\{y,z_{1},z_{2}\}$, and then $V=\{x\}$. 
\qed
\end{example}

\begin{example}
Let $q(z_3)= \exists x\exists y \exists z_1\exists z_2\, R(\underline{x}, y, z_1) \land S(\underline{x,y}, z_2) \land T(\underline{z_1, z_2, z_3}) \land P(\underline{x},y)$. The attack graph of $q$ contains two edges: $R(\underline{x}, y, z_1)$ and $S(\underline{x,y}, z_2)$ both attack $T(\underline{z_1, z_2, z_3})$. Thus, the attack graph of $q$ has two weak components. If we construct $V$ as in the statement of Proposition~\ref{pro:defineX}, we first compute $N=\{y,z_{1},z_{2}\}$, and then $V=\{x\}$.
\qed
\end{example}


The following proposition settles the complexity of checking whether a query belongs to $\cnewclass$.

\begin{proposition}\label{pro:testparsimony}
The following decision problem is in quadratic time:
Given a self-join-free conjunctive query $q(\vec{z})$, decide whether or not $q(\vec{z})$ belongs to $\cnewclass$.
\end{proposition}

\section{$\cnewclass$ Admits Parsimonious Counting}\label{sec:soundness}

In this section, we show the if-direction of Theorem~\ref{the:main}, which is the following theorem.

\begin{theorem}\label{the:mainif}
Every self-join-free conjunctive query in $\cnewclass$ admits parsimonious counting.
\end{theorem}

We use a number of helping lemmas and constructs.
The following lemma says that if~$\vec{x}$ is an id-set of a query $q(\vec{z})$ in $\cnewclass$,
then for a consistent database $\db$,
the answers to $\mycount{\body{q}}{\vec{z}}$ can be obtained by counting the number of distinct $\vec{x}$-values within each~$\vec{z}$-group, while variables not in $\vec{x}\cdot\vec{z}$ can be ignored.

\begin{lemma} \label{lem:onlyOneVal}
Let $q(\vec{z},\vec{x}, \vec{w})$ be a full self-join free conjunctive query, in which notation it is understood that $\vec{z}$, $\vec{x}$ and $\vec{w}$ are disjoint, duplicate-free tuples of variables.
Assume that the query $\makebound{q}{\vec{x}\vec{w}}$ belongs to $\cnewclass$ and that $\vec{x}$ is an id-set for it.
Let $\db$ be a consistent database instance. 
For all tuples $\vec{a}$ and $\vec{b}$ of constants, of arities $\arity{\vec{z}}$ and $\arity{\vec{x}}$ respectively, for all tuples $\vec{c_1}$ and $\vec{c_2}$ of arity $\arity{\vec{w}}$, if $\db \models q(\vec{a}, \vec{b}, \vec{c_1})$ and $\db \models q(\vec{a}, \vec{b}, \vec{c_2})$, then $\vec{c}_1 = \vec{c}_2$.
\end{lemma}

We now present the notion of \emph{optimistic repair}, which was originally introduced by Fuxman~\cite{FuxmanThesis}.
Informally, a repair $\rep$ of a database $\db$ is an optimistic repair with respect to a conjunctive query~$q(\vec{z})$ if every tuple that is an answer to $q(\vec{z})$ on $\db$ is also an answer to $q(\vec{z})$ on $\rep$.
The converse obviously holds true because conjunctive queries are monotone and repairs are subsets of the original database instance.

\begin{definition}[Optimistic repair]\label{def:optiRepair}
Let $q(\vec{x})$ be a conjunctive query. 
Let $\db$ be a database instance. 
We say that a repair $\rep$ of $\db$ is an \emph{optimistic repair} with respect to $q(\vec{x})$ if for every tuple $\vec{a}$ of constants, of arity $\arity{\vec{x}}$, $\db\models q(\vec{a})$ implies $\rep\models q(\vec{a})$ (the converse implication is obviously true).
\end{definition}

The following lemma gives a sufficient condition for the existence of optimistic repairs.

\begin{lemma} \label{lem:hasOptiRepair}
Let $q(\vec{z})$ be a self-join free conjunctive query in $\cnewclass$, and let $\vec{x}$ be a minimal id-set for it.
Let $q'(\vec{z},\vec{x})$ be the query $\makefree{q}{\vec{x}}$.
Let $\db$ be a database instance, and $\vec{c}$ a tuple of constants, of arity~$\arity{\vec{z}}$, such that $\db\cqamodels q(\vec{c})$.
Then, $\db$ has an optimistic repair with respect to $\substitute{q'}{\vec{z}}{\vec{c}}$.
\end{lemma}

We now present the notion of \emph{pessimistic repair}, also borrowed from~\cite{FuxmanThesis}.
Informally, a repair of a database $\db$ is a pessimistic repair with respect to a conjunctive query~$q(\vec{z})$ if every answer to~$q(\vec{z})$ on~$\rep$ is a consistent answer to $q(\vec{z})$ on $\db$.
The converse trivially holds true.

\begin{definition}[Pessimistic repair]\label{def:pessRepair}
Let $q(\vec{x})$ be a conjunctive query. 
Let $\db$ be a database instance. 
We say that a repair $\rep$ of $\db$ is a  pessimistic repair with respect to $q(\vec{x})$ if for every tuple $\vec{a}$ of constants, of arity $\card{\vec{x}}$, if $\rep \models q(\vec{a})$, then $\db\cqamodels q(\vec{a})$.
\end{definition}

The following lemma gives a sufficient condition for the existence of pessimistic repairs.

\begin{lemma} \label{lem:hasPessRepair}
Let $q(\vec{z})$ be a self-join free conjunctive query in $\cnewclass$, and let $\vec{x}$ be a minimal id-set for it.
Let $\db$ be a database instance, and $\vec{c}$ a tuple of constants, of arity~$\arity{\vec{z}}$, such that $\db\cqamodels q(\vec{c})$.
Then, $\db$ has a pessimistic repair with respect to $\substitute{q'}{\vec{z}}{\vec{c}}$.
\end{lemma}

The following example illustrates the preceding constructs and lemmas.

\begin{example}
Let $q(z)= \exists x \exists y \exists v \, R(\underline{x},y)\land S(\underline{y,v},z) \land T(\underline{y},v)$.  Let $\db$ be the following database instance:
\begin{center}\small
\begin{tabular}{ccc}
	\begin{tabular}[t]{c|cc}
		$R$  & $\underline{x}$ & $y$ \\ \cline{2-3}
		& $a_1$ & $b_1$\\\cdashline{2-3}
		& $a_2$ & $b_2$\\\cdashline{2-3}
		& $a_3$ & $b_2$\\\cdashline{2-3}
		& $a_4$ & $b_3$
	\end{tabular}
	&
	\begin{tabular}[t]{c|ccc}
		$S$  & $\underline{y}$ & $\underline{v}$ & $z$\bigstrut\\ \cline{2-4}
		& $b_1$ & $c_1$ & $g_1$\\\cdashline{2-4}
		& $b_2$ & $c_2$ & $g_1$\\
		& $b_2$ & $c_2$ & $g_2$\\\cdashline{2-4}
		& $b_3$ & $c_3$ & $g_2$
	\end{tabular}
	&
	\begin{tabular}[t]{c|cc}
		$T$  & $\underline{y}$ & $v$ \bigstrut\\ \cline{2-3}
		& $b_1$ & $c_1$\\\cdashline{2-3}
		& $b_2$ & $c_2$\\\cdashline{2-3}
		& $b_3$ & $c_3$
	\end{tabular}
\end{tabular} 
\end{center}
Clearly, $\db$ has two repairs, which are $\rep_{1}\defeq\db\setminus\{S(\underline{b_2, b_2}, g_2)\}$ and $\rep_{2}\defeq\db\setminus\{S(\underline{b_2, b_2}, g_1)\}$.

We first determine the answers to $\cqacount{\body{q}}{z}$ on $\db$ in a naive way without using parsimonious counting, but by enumerating repairs.
To this end, let $q^{*}(z,x,y,v)=R(\underline{x},y)\land S(\underline{y,v},z) \land T(\underline{y},v)$.
We have:
\begin{eqnarray*}
q^{*}(\rep_{1}) & = & \{(g_{1},a_{1},b_{1},c_{1}), (g_{1},a_{2},b_{2},c_{2}), (g_{1},a_{3},b_{2},c_{2}), (g_{2},a_{4},b_{3},c_{3})\}\\
q^{*}(\rep_{2}) & = & \{(g_{1},a_{1},b_{1},c_{1}), (g_{2},a_{2},b_{2},c_{2}), (g_{2},a_{3},b_{2},c_{2}), (g_{2},a_{4},b_{3},c_{3})\}
\end{eqnarray*}
The value $g_{1}$ occurs in $3$ tuples of $q^{*}(\rep_{1})$, and in one tuple of $q^{*}(\rep_{2})$.
On the other hand,  $g_{2}$ occurs in one tuple of $q^{*}(\rep_{1})$, and in $3$ tuples of $q^{*}(\rep_{2})$.
It follows that $(g_{1},[1,3])$ and $(g_{2},[1,3])$ are the answers to $\cqacount{\body{q}}{z}$ on $\db$.

It can be verified that $q(z) \in \cnewclass$ with an id-set $\vec{x} = \tuple{x}$. 
We next compute $\cqacount{\body{q}}{z}$ on $\db$ by means of parsimonious counting. 
To this end, let $q'(z,x) = \makefree{q}{x}$, and let $\varphi(z,x)$ be a consistent first-order rewriting for $q'(z,x)$.
If we execute these queries on $\db$, we obtain:\footnote{$\varphi(\db)$ is a shorthand for the set of all tuples $(c,d)$ such that $\db\models\varphi(c,d)$.}
\begin{eqnarray*}
q'(\db) & = & \{(g_{1},a_{1}), (g_{1},a_{2}), (g_{1},a_{3}), (g_{2},a_{2}), (g_{2},a_{3}), (g_{2},a_{4})\}\\
\varphi(\db) & = & \{(g_{1},a_{1}), (g_{2},a_{4})\}
\end{eqnarray*}
As stated in Theorem~\ref{the:mainif}, the set $q'(\db)$ yields the upper bound~$3$ for $g_{1}$ and $g_{2}$,
and the set $\varphi(\db)$ yields the lower bound~$1$ for $g_{1}$ and $g_{2}$.
It is important to understand that parsimonious counting obtains these bounds directly on $\db$, without computing any repair.

We elaborate this example further to illustrate the constructs of optimistic and pessimistic repairs.
We have:
\begin{eqnarray*}
q'(\rep_{1}) & = & \{(g_{1},a_{1}), (g_{1},a_{2}), (g_{1},a_{3}), (g_{2},a_{4})\}\\
q'(\rep_{2}) & = & \{(g_{1},a_{1}), (g_{2},a_{2}), (g_{2},a_{3}), (g_{2},a_{4})\}
\end{eqnarray*}
Note that the consistent answer to $q'(z,x)$ on $\db$ (i.e., the set $\varphi(\db)$ used previously) is equal to $q'(\rep_{1})\cap q'(\rep_{2})=\{(g_{1},a_{1}), (g_{2},a_{4})\}$.
We see that $\rep_{1}$ is an optimistic repair with respect to $\substitute{q'(z,x)}{z}{g_{1}}$, and a pessimistic repair with respect to $\substitute{q'(z,x)}{z}{g_{2}}$. 
On the other hand, $\rep_{2}$ is an optimistic repair with respect to $\substitute{q'(z,x)}{z}{g_{2}}$, and a pessimistic repair with respect to $\substitute{q'(z,x)}{z}{g_{1}}$. 
\qed
\end{example}

We finish this section by a proof of Theorem~\ref{the:mainif}.

\begin{proof}[Proof of Theorem~\ref{the:mainif}]
Let $q(\vec{z}) \in \cnewclass$.
We have to prove that $q(\vec{z})$ admits parsimonious counting.
Since $q(\vec{z})\in\cnewclass$, we can assume an id-set $\vec{x}$ for $q(\vec{z})$.
It suffices to show that conditions~\ref{it:pcone}, \ref{it:pctwo}, and~\ref{it:pcthree} in Definition~\ref{def:pc} are satisfied for this choice of $\vec{x}$.
As in Definition~\ref{def:pc}, let $q'(\vec{z},\vec{x})=\makefree{q}{\vec{x}}$.

Since $q(\vec{z})$ is in $\cnewclass$, it has an acyclic attack graph.
It follows from Theorem~\ref{the:koutriswijsen} that $q(\vec{z})$ has a consistent first-order rewriting.
Thus, condition~\ref{it:pcone} in Definition~\ref{def:pc} is satisfied.
It is known~\cite{DBLP:journals/tods/KoutrisW17} that the attack graph of $q'(\vec{z},\vec{x})$ is a subgraph of the attack graph of $q(\vec{z})$.
Informally, no new attacks are introduced when bound variables are made free.
It follows that $q'(\vec{z},\vec{x})$ has an acyclic attack graph, and therefore, by Theorem~\ref{the:koutriswijsen}, a consistent first-order rewriting.
Thus, condition~\ref{it:pctwo} in Definition~\ref{def:pc} is satisfied.
In the remainder of the proof, we show that condition~\ref{it:pcthree} in Definition~\ref{def:pc} is satisfied.
To this end, let $\db$ be an arbitrary database instance.

Let $\vec{c}$ be a tuple of constants such that $\db\cqamodels q(\vec{c})$.
Let $D$ be the active domain of $\db$.
Let $f$ be a function that maps every subset $\sep$ of $\db$ to the cardinality of the set $\{\vec{a}\in D^{\arity{\vec{x}}}\mid\sep\models q'(\vec{c},\vec{a})\}$.
Clearly, for every repair $\rep$ of $\db$, we have $\rep\subseteq\db$ and hence, since conjunctive queries are monotone, $f(\rep)\leq f(\db)$.
Moreover, since repairs are consistent, 
it follows by Lemma~\ref{lem:onlyOneVal} that for every repair $\rep$ of $\db$, if $(\vec{c},i)$ is an answer to the query $\mycount{\body{q}}{\vec{z}}$ on~$\rep$, then $i=f(\rep)$.

By Lemma~\ref{lem:hasOptiRepair}, we can assume an optimistic repair $\oep$ of $\db$ with respect to $\substitute{q'(\vec{z}, \vec{x})}{\vec{z}}{\vec{c}}$. 
By Definition~\ref{def:optiRepair} of optimistic repair, for every tuple $\vec{a}$ of constants, of arity~$\arity{\vec{x}}$, we have $\oep\models q'(\vec{c},\vec{a})$ if and only if $\db\models q'(\vec{c},\vec{a})$.
It follows $f(\oep)=f(\db)$.
Consequently, for every repair $\rep$ of $\db$, $f(\rep)\leq f(\oep)$.
It follows that for some lower bound~$m$, we have that $(\vec{c},[m,f(\db)]$ is an answer to $\cqacount{\body{q}}{\vec{z}}$ on $\db$.

By Lemma~\ref{lem:hasPessRepair}, we can assume a pessimistic repair $\pep$ of $\db$ with respect to $\substitute{q'(\vec{z},\vec{x})}{\vec{z}}{\vec{c}}$. 
Let $\varphi(\vec{z},\vec{x})$ be a consistent first-order rewriting of $q'(\vec{z},\vec{x})$.
By Definition~\ref{def:pessRepair} of pessimistic repair,
the following hold:
\begin{itemize}
\item 
$\pep\models q(\vec{c},\vec{a})$ if and only if $\db\models \varphi(\vec{c},\vec{a})$.
Therefore, $f(\pep)$ is the cardinality of the set $S\defeq\{\vec{a}\in D^{\arity{\vec{x}}}\mid\db\models \varphi(\vec{c},\vec{a})\}$.
\item
for every repair $\rep$ of $\db$, $f(\pep)\leq f(\rep)$.
\end{itemize}
It follows that for some upper bound~$n$, we have that $(\vec{c},[\card{S},n])$ is an answer to $\cqacount{\body{q}}{\vec{z}}$ on $\db$.
Putting everything together, we obtain that $(\vec{c},[\card{S},f(\db)])$ is an answer to $\cqacount{\body{q}}{\vec{z}}$ on $\db$.
From this, it is correct to conclude that condition~\ref{it:pcthree} in Definition~\ref{def:pc} is satisfied.
This concludes the proof.
\end{proof}

\section{Completeness of $\cnewclass$}\label{sec:completeness}

In this section, we show the only-if-direction of Theorem~\ref{the:main}, which is the following theorem.

\begin{theorem}\label{the:mainonlyif}
Every self-join-free conjunctive query that admits parsimonious counting belongs to $\cnewclass$.
\end{theorem}

The following three lemmas state some properties of queries $q(\vec{z})$ that admit parsimonious counting on some~$\vec{x}$.

\begin{lemma} \label{lem:mustBeInFO}
Let $q(\vec{z})$ be a self-join-free conjunctive query.
If $q(\vec{z})$ admits parsimonious counting,
then the attack graph of $q(\vec{z})$ is acyclic.
\end{lemma}
\begin{proof}
Proof by contraposition.
If the attack graph of $q(\vec{z})$ is cyclic, then by Theorem~\ref{the:koutriswijsen}, $q(\vec{z})$ has no consistent first-order rewriting, and therefore does not admit parsimonious counting.
\end{proof}

\begin{lemma} \label{lem:strongImpliesMultipleValuations}
Let $q(\vec{z})$ be a self-join-free conjunctive query.
Let $\vec{x}$ be a (possibly empty) sequence of bound variables of $q(\vec{z})$.
If $q(\vec{z})$ admits parsimonious counting on~$\vec{x}$,
then the attack graph of  $q'(\vec{z}, \vec{x})$ has no strong attack.
\end{lemma}

\begin{lemma} \label{lem:freeBadComponentImpliesMultipleValuations}
Let $q(\vec{z})$ be self-join-free conjunctive query whose attack graph is acyclic.
Let~$\vec{x}$ be a (possibly empty) sequence of bound variables of $q(\vec{z})$.
If $q(\vec{z})$ admits parsimonious counting on~$\vec{x}$,
then $\vec{x}$ satisfies condition~\ref{it:source} in Definition~\ref{def:newclass}.
\end{lemma}


The following two lemmas, and their corollary, concern condition~\ref{it:separation} in Definition~\ref{def:newclass}.

\begin{lemma}\label{lem:optimistic}
Let $q(\vec{z})$ be a self-join-free conjunctive query.
Let $\vec{x}$ be a (possibly empty) sequence of bound variables of $q(\vec{z})$, and let $q'(\vec{z},\vec{x})=\makefree{q}{\vec{x}}$.
Let $\vec{c}$ a tuple of constants of arity~$\arity{\vec{z}}$.
If $q(\vec{z})$ admits parsimonious counting on~$\vec{x}$,
then for every database instance $\db$, if $\db\cqamodels q(\vec{c})$, then $\db$ has an optimistic repair with respect to $\substitute{q'}{\vec{z}}{\vec{c}}$.
\end{lemma}
\begin{proof}
Assume that $q(\vec{z})$ admits parsimonious counting on~$\vec{x}$.
Let $\db$ be a database instance such that $\db\cqamodels q(\vec{c})$.
Let $(\vec{c}, [m,n])$ be an answer to $\cqacount{\body{q}}{\vec{z}}$ on $\db$. 
Define
\begin{equation}
\calD\defeq\{\vec{d}\in D^{\arity{\vec{x}}}\mid\db\models q'(\vec{c},\vec{d})\}, 
\end{equation}
where $D$ be the active domain of $\db$.
By our hypothesis that $q(\vec{z})$ admits parsimonious counting on $\vec{x}$,
it follows by condition~\ref{it:pcthree} in Definition~\ref{def:pc} that
\begin{equation}\label{eq:nd}
n=\card{\calD}.
\end{equation}
By Definition~\ref{def:cqacnt}, we can assume a repair $\rep$ of $\db$ such that $(\vec{c}, n)$ is an answer to $\mycount{\body{q}}{\vec{z}}$ on $\rep$.
Since $\rep$ is consistent, we have that $(\vec{c}, [n,n])$ is an answer to $\cqacount{\body{q}}{\vec{z}}$ on $\rep$. 
Define
\begin{equation}
\calR\defeq\{\vec{d}\in D^{\arity{\vec{x}}}\mid\rep\models q'(\vec{c},\vec{d})\}.
\end{equation}
By our hypothesis that $q(\vec{z})$ admits parsimonious counting on $\vec{x}$,
it follows by condition~\ref{it:pcthree} in Definition~\ref{def:pc} that
\begin{equation}\label{eq:nr}
n=\card{\calR}.
\end{equation}
Since conjunctive queries are monotone and $\rep\subseteq\db$, it follows $\calR\subseteq\calD$.
Since $\card{\calR}=\card{\calD}$ by~\eqref{eq:nd} and~\eqref{eq:nr}, it follows $\calR=\calD$.
From $\calD\subseteq\calR$,
it follows that $\rep$ is an optimistic repair with respect to $\substitute{q'(\vec{z}, \vec{x})}{\vec{z}}{\vec{c}}$.
\end{proof}

\begin{lemma} \label{lem:freeBadPathImpliesNoOptiRepair}
Let $q(\vec{z})$, $\vec{x}$, $q'(\vec{z},\vec{x})$, and $\vec{c}$ be as in the statement of Lemma~\ref{lem:optimistic}.
Assume that~$\vec{x}$ violates condition~\ref{it:separation} in Definition~\ref{def:newclass}.
Then, there exists a database $\db$ such that $\db\cqamodels q(\vec{c})$, but $\db$ has no optimistic repair with respect to $\substitute{q'}{\vec{z}}{\vec{c}}$.
\end{lemma}

\begin{corollary}\label{cor:optimistic}
Let $q(\vec{z})$ be a self-join-free conjunctive query.
Let $\vec{x}$ be a sequence of distinct bound variables of $q(\vec{z})$, and let $q'(\vec{z},\vec{x})=\makefree{q}{\vec{x}}$.
If $q(\vec{z})$ admits parsimonious counting on~$\vec{x}$,
then $\vec{x}$ satisfies condition~\ref{it:separation} in Definition~\ref{def:newclass}.
\end{corollary}
\begin{proof}
Immediately from Lemmas~\ref{lem:optimistic} and~\ref{lem:freeBadPathImpliesNoOptiRepair}.
\end{proof}

Finally, we need the following result.

\begin{lemma}\label{lem:strongImpliesStrong}
Let $q(\vec{z})$ be a self-join-free conjunctive query.
Let $\vec{x}$ be a sequence of distinct bound variables of $q(\vec{z})$.
Let $q'(\vec{z},\vec{x})=\makefree{q}{\vec{x}}$.
Assume that $\vec{x}$ satisfies condition~\ref{it:separation} in Definition~\ref{def:newclass}.
If the attack graph of $q(\vec{z})$ has a strong attack from an atom $R$ to an atom $S$,
then the attack graph of $q'(\vec{z},\vec{x})$ has a strong attack from $R$ to $S$.
\end{lemma}

Before giving a proof of Theorem~\ref{the:mainonlyif}, we illustrate the preceding results with an example.

\begin{example}
Let $q(z) = \exists x\exists y\, R(\underline{x},z,y)\land S(\underline{y},x)\land T(\underline{y},x)$.
We will argue that $q(z)$ is not in $\cnewclass$,  
and then illustrate that it does not admit parsimonious counting. 

The only edges in the attack graph of $q$ are $(R,S)$ and $(R,T)$.
Assume for the sake of contradiction that $q\in\cnewclass$.
Then, following Proposition~\ref{pro:defineX}, the minimal id-set of $q(\vec{z})$ is $\tuple{}$.
However, since $\fdset{q(z)}\equiv\{\fd{x}{y}, \fd{y}{x}, \fd{\emptyset}{z}\}$, condition~\ref{it:source} in Definition~\ref{def:newclass} is violated for $\vec{x}=\tuple{}$.
We conclude by contradiction that $q\notin\cnewclass$.
%

We now argue, without using Theorem~\ref{the:mainonlyif}, that $q(z)$ does not admit parsimonious counting.
Conditions~\ref{it:pcone} and~\ref{it:pctwo} in Definition~\ref{def:pc} of parsimonious counting are satisfied for every choice of $\vec{x}$ in $\{\tuple{}$, $\tuple{x}$, $\tuple{y}$, $\tuple{x,y}\}$. 
However, we will show that condition~\ref{it:pcthree} is not satisfied.
To this end, let $\vec{x}$ be a sequence of bound variables of $q(z)$. Let $q'(z, \vec{x}) = \makefree{q}{\vec{x}}$.
First, suppose that $\vec{x}\in\{\tuple{x}, \tuple{x,y}\}$.  
Consider the following database instance $\db$:
\begin{small}
$$
\begin{array}{ccc}
	\begin{array}[t]{c|ccc}
		R  & \underline{x} & z & y \\ \cline{2-4}
		& a & d & e\\\cdashline{2-4}
		& b & d & e\\\cdashline{2-4}
		& c & d & f
	\end{array}
	&
	\begin{array}[t]{c|cc}
		S  & \underline{y} & x \bigstrut\\ \cline{2-3}
		& e & a\\
		& e & b\\\cdashline{2-3}
		& f & c
	\end{array}
	&
	\begin{array}[t]{c|cc}
		T  & \underline{y} & x \bigstrut\\ \cline{2-3}
		& e & a\\
		& e & b\\\cdashline{2-3}
		& f & c
	\end{array}
\end{array}
$$
\end{small}
We have that $(d, [1,2])$ is an answer to $\cqacount{\body{q}}{z}$, but it can be easily verified that $\card{q'(\db)}=3$, which is distinct from the upper bound~$2$. 

Assume next that $\vec{x}\in\{\tuple{y}, \tuple{x,y}\}$. 
Consider the following database instance $\db$:
\begin{small}
$$
\begin{array}{ccc}
	\begin{array}[t]{c|ccc}
		R  & \underline{x} & z & y \\ \cline{2-4}
		& a & d & e\\
		& a & d & f\\\cdashline{2-4}
		& b & d & g
	\end{array}
	&
	\begin{array}[t]{c|cc}
		S  & \underline{y} & x \bigstrut\\ \cline{2-3}
		& e & a\\\cdashline{2-3}
		& f & a\\\cdashline{2-3}
		& g & b
	\end{array}
	&
	\begin{array}[t]{c|cc}
		T  & \underline{y} & x \bigstrut\\ \cline{2-3}
		& e & a\\\cdashline{2-3}
		& f & a\\\cdashline{2-3}
		& g & b
	\end{array}
\end{array}
$$
\end{small}
Now we have that $(d, [2,2])$ is an answer to $\cqacount{\body{q}}{z}$, but $\card{q'(\db)}=3$.

The only remaining case to be considered is $\vec{x}=\tuple{}$. 
In that case $q'=q$. Consider the following database instance $\db$:
\begin{small}
$$
\begin{array}{ccc}
	\begin{array}[t]{c|ccc}
		R  & \underline{x} & z & y \\ \cline{2-4}
		& a & d & e\\\cdashline{2-4}
		& b & d & f
	\end{array}
	&
	\begin{array}[t]{c|cc}
		S  & \underline{y} & x \bigstrut\\ \cline{2-3}
		& e & a\\\cdashline{2-3}
		& f & b
	\end{array}
	&
	\begin{array}[t]{c|cc}
		T  & \underline{y} & x \bigstrut\\ \cline{2-3}
		& e & a\\\cdashline{2-3}
		& f & b
	\end{array}
 \end{array}
 $$
\end{small}
Since $\db$ is a consistent database instance, the only repair of $\db$ is $\db$ itself. We have that $(d,[2,2])$ is an answer to $\cqacount{\body{q}}{z}$ on $\db$.
It can be easily verified that $\card{q'(\db)}=1$, which is distinct from the upper bound~$2$.

Finally, we claim (without proof) that \problem{2-DIMENSIONAL MATCHING} (\problem{2DM}) can be first-order reduced to computing $\cqacount{\body{q}}{z}$.
Therefore, since \problem{2DM} is $\NL$-hard~\cite{DBLP:journals/siamcomp/ChandraSV84}, $q(z)$ cannot admit parsimonious counting under standard complexity assumptions.  
\qed
\end{example}

We can now give the proof of Theorem~\ref{the:mainonlyif}.

\begin{proof}[Proof of Theorem~\ref{the:mainonlyif}]
Assume that $q(\vec{z})$ admits parsimonious counting.
Then, $q(\vec{z})$ has a tuple $\vec{x}$ of bound variables such that for the query  $q'(\vec{z},\vec{x})\defeq\makefree{q}{\vec{x}}$,
the conditions~\ref{it:pcone}, \ref{it:pctwo}, and \ref{it:pcthree} in Definition~\ref{def:pc} are satisfied.
From conditions~\ref{it:pcone} and~\ref{it:pctwo}, it follows by Theorem~\ref{the:koutriswijsen} that $q(\vec{z})$ and $q'(\vec{z}, \vec{x})$ have acyclic attack graphs.
By Lemma~\ref{lem:freeBadComponentImpliesMultipleValuations}, condition~\ref{it:source} in Definition~\ref{def:newclass} is satisfied for~$\vec{x}$.
By Corollary~\ref{cor:optimistic}, condition~\ref{it:separation} in Definition~\ref{def:newclass} is satisfied by $\vec{x}$.
By Lemma~\ref{lem:strongImpliesMultipleValuations}, the attack graph of $q'(\vec{z}, \vec{x})$ has no strong attack. 
By Lemma~\ref{lem:strongImpliesStrong}, it is now correct to conclude that the attack graph of $q(\vec{z})$ has no strong attack either, and thus condition~\ref{it:noCyclicOrStrongAttacks} in Definition~\ref{def:newclass} is satisfied.
Since we have shown that $q(\vec{z})$ satisfies all conditions in Definition~\ref{def:newclass}, we conclude $q(\vec{z})\in\cnewclass$.
\end{proof}

\section{Comparison with $\cforest$}\label{sec:cforest}

In this section, we introduce $\cforest$ and show $\cforest\subsetneq\cnewclass$ without making use of Theorem~\ref{the:fuxman}.
Theorem~\ref{the:fuxman} then follows by Theorem~\ref{the:mainif}.

\begin{definition}[$\cforest$]
Let $q(\vec{z})$ be a self-join-free conjunctive query. The \emph{Fuxman graph} of $q$ is a directed graph whose vertices are the atoms of~$q$.
There is a directed edge from an atom $R$ to an atom $S$ if $R\neq S$ and $\notkey{R}$ contains a bound variable that also occurs in~$S$.
The class $\cforest$ contains all (and only)  self-join free conjunctive queries $q(\vec{z})$ whose Fuxman graph is a directed forest satisfying, for every directed edge from $R$ to $S$, $\key{S}\setminus\free{q}\subseteq\notkey{R}$.
\end{definition}

\begin{theorem}\label{the:cforest}
$\cforest$ is a strict subset of $\cnewclass$.
\end{theorem}

\section{Conclusion and Future Work}\label{sec:conclusion}

In his PhD thesis, Fuxman~\cite{FuxmanThesis} defined a syntactically  restricted class of self-join-free  conjunctive queries, called $\cforest$, and showed that for every query in $\cforest$, consistent answers are first-order computable, and range-consistent answers are computable in first-order logic followed by a simple aggregation step.
Our notion of ``parsimonious counting'' captures the latter computation for counting.
Later, Koutris and Wijsen~\cite{DBLP:journals/tods/KoutrisW17} syntactically characterized the class of \emph{all} self-join-free conjunctive queries with a consistent first-order rewriting, which strictly includes~$\cforest$.
However, it remained an open problem to syntactically characterize the class of \emph{all} self-join-free conjunctive queries that admit parsimonious counting.
In this paper, we determined the latter class, named it $\cnewclass$, and showed that it strictly includes $\cforest$.

We now list some open problems for future research.
In Definition~\ref{def:pc} of parsimonious counting, we required that $q(\vec{z})$ has a consistent first-order rewriting.
It is known~\cite{DBLP:journals/mst/KoutrisW21} that there are self-join-free conjunctive queries, without consistent first-order rewriting, that have a consistent rewriting in Datalog. 
We could relax Definition~\ref{def:pc} by requiring the existence of a consistent rewriting in Datalog, rather than in first-order logic.
It is an open question to syntactically characterize the self-join-free conjunctive queries that admit parsimonious counting under such a relaxed definition. 

Another open question is to characterize the complexity of $\cqacount{q(\vec{z},\vec{w})}{\vec{z}}$ for every full self-join-free conjunctive query $q$ and choice of free variables $\vec{z}$. 
It is easily verified that the complexity of computing the answers to $\cqacount{q(\vec{z},\vec{w})}{\vec{z}}$ is higher than computing the consistent answers to $q'(\vec{z})\defeq\makebound{q}{\vec{w}}$ (because of the lower bound in range semantics).
It remains an open question to characterize this complexity if $q'(\vec{z})$ is not in $\cnewclass$, even if it has a consistent first-order rewriting.

The notion of parsimonious counting does not require conjunctive queries to be self-join-free.
An ambitious open problem is to syntactically characterize the class of all (i.e., not necessarily self-join-free) conjunctive queries that admit parsimonious counting.
This problem is largely open, because it is already a notorious open problem to syntactically characterize the class of conjunctive queries that have a consistent first-order rewriting.

Another open question is to extend the results in the current paper to other aggregation operators than COUNT, including MAX, MIN, SUM, and AVG.

\bibliography{dblp}

\appendix

\section{Helping Constructs and Lemmas}



We recall some definitions from graph theory. 
In a directed acyclic graph (DAG), a vertex with zero indegree is called a \emph{source}. 
If there is a directed path from $u$ to $v$ in a DAG, with $u\neq v$, then $v$ is called a \emph{descendant} of~$u$, and $u$ an \emph{ancestor} of~$v$.




The following lemma states that whenever $R$ and $S$ are distinct sources that are weakly connected in a DAG,
then there is a sequence of sources that starts with~$R$ and ends with~$S$,
such that every two adjacent vertices in the sequence have a descendant in common.

\begin{lemma}[Sources of a graph are linked by sources]\label{lem:rootsLinkedByRoots}
	Let $G=(V,E)$ be a DAG. 
	Let $R,S \in V$ be distinct sources that are weakly connected. 
	There is a sequence $\tuple{H_1, \ldots, H_n}$ of sources such that $H_{1}=R$, $H_{n}=S$, and every two adjacent sources in the sequence have a common descendent. 
\end{lemma}

\begin{proof}
    Since $R$ and $S$ are weakly connected, there is a sequence $\tuple{F_1, \ldots, F_n}$ of atoms in~$V$ such that $F_1 = R$, $F_n=S$ and for every $i \in \{1, \ldots, n-1\}$, either $(F_{i}, F_{i+1}) \in E$ or $(F_{i+1}, F_{i}) \in E$.
    We construct a new sequence by repeating the following step as long as possible: replace some contiguous subsequence $\tuple{I,\ldots,J}$ with $\tuple{I,J}$ if $I$ is a descendant or an ancestor of $J$.
    Let the new sequence be $\tuple{G_1, \ldots, G_m}$.
    By construction, for all $i\in\{1,\ldots,m-2\}$,
    \begin{itemize}
        \item if $G_{i}$ is a descendant of $G_{i+1}$, then $G_{i+1}$ is an ancestor of $G_{i+2}$; and
        \item if $G_{i}$ is an ancestor of $G_{i+1}$, then $G_{i+1}$ is a descendant of $G_{i+2}$.
    \end{itemize}
    It is easily verified that $G_{1}=R$, $G_{m}=S$, and that $m$ is an odd number.
    By construction, for every $i\in\{1,\dots,m\}$ such that $i$ is even, we have that $G_{i}$ is a common descendant of $G_{i-1}$ and $G_{i+1}$. 
    We now change the sequence $\tuple{G_{1},\dots,G_{m}}$ as follows: for every $j\in\{1,\dots,m\}$ such that $j$ is odd, if $G_{j}$ is not a source, replace $G_{j}$ by a source that is an ancestor of~$G_{j}$ (such a source obviously exists in a DAG).
    It is easily verified that if we now omit the vertices that are not sources (i.e., $G_{2},G_{4},\ldots, G_{m-1}$), we obtain a sequence of sources in which every two adjacent sources have a descendant in common. 
\end{proof}

The following lemma states that sequential proofs provide a sound and complete characterization of logical implication.

\begin{lemma}\label{lem:soundcomplete}
	Let $q$ be a self-join free conjunctive query. 
        Let $Z\subseteq\vars{q}$ and $w\in\vars{q}$.
        Then the following are equivalent:
        \begin{enumerate}
            \item $\fdset{q}\models\fd{Z}{w}$; and
            \item there is a sequence $\tuple{F_1, \ldots, F_n}$ of atoms that is a sequential proof of $\fdset{q}\models\fd{Z}{w}$. 
        \end{enumerate}
\end{lemma}
\begin{proof}
Sequential proofs mimic a standard algorithm for logical implication of functional dependencies; see for example~\cite[Algorithm~8.2.7]{DBLP:books/aw/AbiteboulHV95}.
\end{proof}



\begin{lemma}[Attacks with same endpoint]\label{lem:sameEndpoint}
	Let $q$ be a self-join-free conjunctive query.
	Let $R$ and $S$ be two distinct atoms that both attack a same atom.
 If $\notattacks{R}{S}{q}$, then either $\attacks{S}{R}{q}$ or $\key{S}\subseteq\keycl{R}{q}$.
\end{lemma}
\begin{proof}
	Assume $\notattacks{R}{S}{q}$.
	There is $T$ such that $\attacks{R}{T}{q}$ and $\attacks{S}{T}{q}$.
	There is a sequence $\tuple{x_{0},x_{1},\ldots,x_{n}}$ ($n\geqslant 0$) of bound variables not in $\keycl{R}{q}$ such that $x_{0}\in\notkey{R}$, $x_{n}\in\vars{T}$, and every two adjacent variables occur together in some atom of $q$.
	Likewise, there is a sequence $\tuple{y_{0},y_{1},\ldots,y_{m}}$ ($m\geqslant 0$) of bound variables not in $\keycl{S}{q}$ such that $y_{0}\in\notkey{S}$, $y_{m}\in\vars{T}$, and every two adjacent variables occur together in some atom of $q$.
	Clearly, for every $i\in\{0,1,\ldots,m\}$, $\attacks{S}{y_{i}}{q}$.
	In the sequence $\tuple{x_{0},x_{1},\ldots,x_{n},y_{m},y_{m-1},\ldots,y_{0}}$, it holds that $x_{0}\in\notkey{R}$, $y_{0}\in\notkey{S}$, and every two adjacent variables occur together in some atom of $q$. By our hypothesis that $\notattacks{R}{S}{q}$, there is $i\in\{0,\ldots,m\}$ such that $y_{i}\in\keycl{R}{q}$.
	By Lemma~\ref{lem:soundcomplete}, there exists a shortest sequence $\tuple{H_{1},H_{2},\ldots,H_{\ell}}$ that is a sequential proof of $\fdset{q\setminus\{R\}}\models\fd{\key{R}}{y_{i}}$.   
Two cases are possible.
 \begin{description}
 \item[Case that $\ell=0$.]
 Then $y_{i}\in\key{R}$. From $\attacks{S}{y_{i}}{q}$, it follows $\attacks{S}{R}{q}$.
 \item[Case that $\ell>0$.]	
	If $S\in\{H_{1},\ldots,H_{\ell}\}$, then by Lemma~\ref{lem:soundcomplete} it is correct to conclude $\key{S}\subseteq\keycl{R}{q}$, and the desired result obtains.
	Assume $S\notin\{H_{1},\ldots,H_{\ell}\}$ from here on.
	For technical reasons, define $H_{0}\defeq R$.
	We show that for every $k\in\{1,\dots,\ell\}$, the following holds true:
	\begin{quote}
		\emph{Back Property:}
		if $\attacks{S}{H_{k}}{q}$, then there is $g<k$ such that $\attacks{S}{H_{g}}{q}$. 
	\end{quote}
	To this end, assume $\attacks{S}{H_{k}}{q}$ with $k\geqslant 1$. 
	By definition of attacks, there is a bound variable $u\in\key{H_{k}}$ such that $\attacks{S}{u}{q}$.
	If $u\in\vars{R}$, then the desired result obtains because we let $H_{0}=R$.
	Assume $u\notin\vars{R}$ from here on.
	Then there exists $g\in\{1,2,\ldots,k-1\}$ such that $u\in\notkey{H_{g}}$.
	Informally, $H_{g}$ is the atom that introduces $u$ in $\keycl{R}{q}$.
	Since $\attacks{S}{u}{q}$ and $H_{g}\neq S$, it follows $\attacks{S}{H_{g}}{q}$.
	This concludes the proof of the \emph{Back Property}.
	
	Since our sequential proof is as short as possible, $y_{i}\in\notkey{H_{\ell}}$. 
	From $\attacks{S}{y_{i}}{q}$ and $H_{\ell}\neq S$, it follows $\attacks{S}{H_{\ell}}{q}$.
	By repeated application of the \emph{Back Property}, we obtain $\attacks{S}{H_{0}}{q}$ with $H_{0}=R$, as desired.
	This concludes the proof.\qedhere
 \end{description}
\end{proof}


\begin{lemma}[Sources determine one another] \label{lem:rootKeysDetermineOtherRootsKeys}
	Let $q$ be a self-join free conjunctive query such that $q$'s attack graph is acyclic. 
	Let $R$ and $S$ be unattacked atoms that are weakly connected in $q$'s attack graph.
	Then, $\fdset{q}\models \fd{\key{R}}{\key{S}}$.
\end{lemma}
\begin{proof}
    It follows from~\cite[Lemma~3.5]{DBLP:journals/tods/KoutrisW17}  that if an attack graph is acyclic, then it is transitive.
	By Lemma~\ref{lem:rootsLinkedByRoots}, we can assume a sequence of unattacked atoms $\tuple{H_1, \ldots, H_n}$ such that $H_{1}=R$, $H_{n}=S$, and every two adjacent atoms in the sequence attack a same atom. 
	From Lemma~\ref{lem:sameEndpoint}, it follows that for every $i\in\{1,\ldots,n-1\}$, $\fdset{q}\models \fd{\key{H_{i}}}{\key{H_{i+1}}}$. By Armstrong's transitivity axiom, $\fdset{q}\models\fd{\key{R}}{\key{S}}$.
\end{proof}

\begin{lemma}
	\label{lem:notAttackedKeyDetermineComponent}
	Let $q$ be a self-join free conjunctive query with an acyclic attack graph containing no strong attacks. 
	For every unattacked atom $R$ in $q$'s attack graph, if $w$ is a variable that occurs in an atom that is weakly connected to $R$,
	then  $\fdset{q} \models \fd{\key{R}}{w}$.
\end{lemma}

\begin{proof}
Let $R$ be an unattacked atom in $\body{q}$.
    Let $S$ be a atom in $\body{q}$ that is weakly connected to $R$, and $w \in \vars{S}$.
	Clearly, $\fdset{q} \models \fd{\key{S}}{w}$.
	It follows from~\cite[Lemma~3.5]{DBLP:journals/tods/KoutrisW17}  that if an attack graph is acyclic, then it is transitive.
	Since $q$'s attack graph is acyclic, we can assume an unattacked atom $S'$ such that either $S=S'$ or $\attacks{S'}{S}{q}$.
	Since the attack graph of $q$ contains no strong attacks, it follows $\fdset{q} \models \fd{\key{S'}}{\key{S}}$.
	By Lemma~\ref{lem:rootKeysDetermineOtherRootsKeys}, $\fdset{q} \models \fd{\key{R}}{\key{S'}}$. 
	By Armstrong's transitivity axiom, we obtain $\fdset{q} \models \fd{\key{R}}{w}$.
\end{proof}

\section{Proofs for Section \ref{sec:newclass}}

We present a number of helping lemmas.

\begin{lemma}\label{lem:xMinimal}
Let $q(\vec{z})$ be self-join-free conjunctive query.
Let $\vec{x}$ be a minimal id-set for~$q(\vec{z})$.
Then, for every variable $v$ in $\vec{x}$, we have $\fdset{q}\not\models\fd{\emptyset}{v}$.
Consequently, a minimal id-set contains no frozen variables.
\end{lemma}
\begin{proof}
Assume for the sake of contradiction that there is a variable $v$ in $\vec{x}$ such that $\fdset{q}\models\fd{\emptyset}{v}$.
It can be easily verified that $\vec{x}\setminus\{v\}$ still satisfies conditions~\ref{it:source} and~\ref{it:separation} in Definition~\ref{def:newclass}, contradicting that $\vec{x}$ is minimal.
\end{proof}

\begin{lemma}\label{lem:witnessNotFrozen}
Let $q(\vec{z})$ be a self-join-free conjunctive query.
If a variable $v$ of~$q$ is attacked in~$q$, then $v\notin\frozen{q}$.
\end{lemma}
\begin{proof}
Assume $\attacks{R}{v}{q}$ for some atom $R$.
Assume for the sake of contradiction that $v\in\frozen{q}$.
Then there is a sequential proof of $\fdset{q}\models\fd{\emptyset}{v}$ such that no atom in the sequential proof attacks $v$.
Since $\attacks{R}{v}{q}$, we have that $R$ is not used in the sequential proof.
Consequently, $\fdset{q\setminus\{R\}}\models\fd{\emptyset}{v}$, and hence $v\in\keycl{R}{q}$.
Consequently, $\notattacks{R}{v}{q}$, a contradiction.
\end{proof}

\begin{lemma}\label{lem:xIsKeyofNA}
Let $q(\vec{z})$ be a self-join-free conjunctive query with an acyclic attack graph.
Let $\vec{x}$ be a minimal id-set for $q(\vec{z})$.
For every variable $x$ of $\vec{x}$, the following hold:
\begin{enumerate}[label=(\alph*)]
    \item\label{it:idsetpk} if $x$ occurs in an atom $R$ of $q(\vec{z})$, then $x$ occurs at a primary-key position of~$R$; 
    \item\label{it:idsetunattacked} $x$ is unattacked; and
    \item\label{it:idsetsource} $x$ occurs at a primary-key position in an unattacked atom of $q$.
\end{enumerate}
\end{lemma}
\begin{proof}
\framebox{Proof of~\ref{it:idsetpk}}
Assume for the sake of contradiction that for some atom $R$ in $q(\vec{z})$,
$\vec{x}$ contains some variable $x$ of $\notkey{R}\setminus\vec{z}$.
Since $x$ is a bound variable, the query graph of $q(\vec{z})$ contains an empty path from~$x$ to~$x$,
which is a path from a variable in $\notkey{R}$ to a variable of $\vec{x}$ that uses no variable of $\key{R}$.
From Lemma~\ref{lem:xMinimal}, it follows that $x$ is not frozen.
It follows that condition~\ref{it:separation} in Definition~\ref{def:newclass} is violated, a contradiction.

\framebox{Proof of~\ref{it:idsetunattacked}}
Assume for the sake of contradiction that some variable $x$ of $\vec{x}$ is attacked.
There is an atom $F$ in $q$ that attacks~$x$, with a witness $\tuple{x_{1}, \ldots, x_{n}}$ of the form~\eqref{eq:witness} in the definition of attack graphs, where $x_{n}=x$.
It is easily verified that the witness is a path from $x_{1}$ to $x$ in the query graph of~$q(\vec{z})$.
Since $\key{F}\subseteq\keycl{F}{q}$, no $x_{i}$ in this sequence occurs in $\key{F}$.
By Lemma~\ref{lem:witnessNotFrozen}, no $x_{i}$ is frozen.
It follows that condition~\ref{it:separation} in Definition~\ref{def:newclass} is violated, a contradiction.

\framebox{{Proof of~\ref{it:idsetsource}}}
By item~\ref{it:idsetpk}, we can assume an atom $R$ such that $x\in\key{R}$.
If $R$ is unattacked, then the desired result obtains. Assume from here on that $R$ is attacked by an atom $S$ in $\body{q}$.
By definition of an attack graph, there is a witness $\tuple{v_{1}, \ldots, v_{n}}$ of this attack such that $v_{1} \in\notkey{S}$, $v_{n}\in\key{R}$, and for every $i \in \{1, \ldots, n\}$, $v_{i} \not \in \keycl{S}{q}$ and therefore $v_{i}\not\in\key{S}$. 
By Lemma~\ref{lem:witnessNotFrozen}, no $v_{i}$ is frozen.
From Lemma~\ref{lem:xMinimal}, it follows that $x$ is not frozen.
If $x\notin\key{S}$, then $\tuple{v_{1}, \ldots, v_{n}, x}$ is a (not necessarily simple) path in the query graph from a variable in $\notkey{S}$ to a variable of $\vec{x}$ that uses no variable of $\key{S}$, contradicting condition~\ref{it:separation} in Definition~\ref{def:newclass}.
We conclude by contradiction that $x\in\key{S}$.
If $S$ is unattacked, then the desired result obtains. 
Otherwise there is an atom $T$ that attacks~$S$, and we can repeat the same reasoning as before. 
The same reasoning can however not be applied forever, since the attack graph of $q$ is finite and contains no cycles. 
Therefore, at some point we will find an unattacked atom whose primary key contains~$x$.
\end{proof}

We can now give the proof of Proposition~\ref{pro:defineX}.

\begin{proof}[Proof of Proposition~\ref{pro:defineX}]
\framebox{$V\subseteq\vars{\vec{x}}$.}
Let $v\in V$.
By the construction of $V$, we can assume an unattacked atom~$R$ such that $v$ is a bound variable in $\key{R}\setminus N$.
By condition~\ref{it:source} in Definition~\ref{def:newclass}, there is an unattacked atom~$S$ such that $S$ is weakly connected to~$R$ and $\fdset{q}\models\fd{\vec{x}}{\key{S}}$.
By Lemma~\ref{lem:notAttackedKeyDetermineComponent}, $\fdset{q}\models\fd{\key{S}}{v}$.
From $\fdset{q}\models\fd{\vec{x}}{\key{S}}$ and $\fdset{q}\models\fd{\key{S}}{v}$, it follows $\fdset{q}\models\fd{\vec{x}}{v}$.
Since $v\not\in N$, $\fdset{q}\models\fd{\vec{x}}{v}$ implies that $v$ belongs to~$\vec{x}$.

\framebox{$\vars{\vec{x}}\subseteq V$.}
Let $x$ be an arbitrary variable in $\vec{x}$.
By item~\ref{it:idsetsource} in Lemma~\ref{lem:xIsKeyofNA}, we can assume an unattacked atom~$R$ such that $x\in\key{R}$.
By item~\ref{it:idsetpk} in Lemma~\ref{lem:xIsKeyofNA}, $x\not\in N$.
Then, by construction of $V$, we have $x\in V$.

\framebox{Proof of \ref{it:defineXtwo}.}
Assume $R,S$ are unattacked atoms that are weakly connected in $q(\vec{z})$'s attack graph.
We show $\key{R}\cap\vec{x}\subseteq\key{S}$.
The desired result then follows by symmetry.
By Lemma~\ref{lem:rootsLinkedByRoots},
there is a sequence of weakly connected, unattacked atoms $\tuple{H_{1},H_{2},\ldots,H_{n}}$ such that $H_{1}=R$, $H_{n}=S$, and every two adjacent atoms have a descendant in common.
We will show that $\key{H_{1}}\cap\vec{x}\subseteq\key{H_{2}}$.
We can assume that $T$ is a descendant shared by $H_{1}$ and $H_{2}$.
Since the attack graph of $q$ is acyclic (because $q\in\cnewclass$), it is transitive by~\cite[Lemma~3.5]{DBLP:journals/tods/KoutrisW17}.
Assume $x \in\key{H_{1}}\cap{\vec{x}}$.
Let $\tuple{x_{1}, \ldots, x_{n}}$ be a witness that $\attacks{H_{1}}{T}{q}$. 
It follows that $H_{1}$ attacks $x_{i}$ for every $i \in \{1, \ldots, n\}$.
Let $\tuple{y_{1}, \ldots, y_m}$ be a witness that $\attacks{H_{2}}{T}{q}$.
We have $y_{1},\ldots,y_{m}\notin\key{H_{2}}$.
By Lemma~\ref{lem:witnessNotFrozen}, no variable in $\{x_{1},\dots,x_{n}, y_{1},\dots,y_{m}\}$ is frozen.
Then $\tuple{x,x_{1}, \ldots, x_{n}, y_m, \ldots, y_{1}}$ is a path in the query graph of $q$ between $x$ and $y_{1}$.
For every $i\in\{1,\dots,n\}$, $x_{i}\notin\key{H_{2}}$, or else $\attacks{H_{1}}{H_{2}}{q}$, a contradiction.
Assume towards a contradiction that $x\notin\key{H_{2}}$.
The reverse path is a path from the variable~$y_{1}\in\notkey{H_{2}}$ to a variable of $\vec{x}$ that uses no variable of $\key{H_{2}}\cup\frozen{q}$.
By condition~\ref{it:separation} in Definition~\ref{def:newclass}, $\vec{x}$ is not an id-set, a contradiction.
We conclude by contradiction that $x\in\key{H_{2}}$.

By repeating the same reasoning, we obtain $\key{H_{2}}\cap\vec{x}\subseteq\key{H_{3}}$,  $\key{H_{3}}\cap\vec{x}\subseteq\key{H_{4}}$,\dots\ It follows  $\key{H_{1}}\cap\vec{x}\subseteq\key{H_{n}}$.
\end{proof}

Finally, we give the proof of Proposition~\ref{pro:testparsimony}.

\begin{proof}[Proof of Proposition~\ref{pro:testparsimony}]
Let $q(\vec{z})$ be a self-join-free conjunctive query.
It is possible, in quadratic time in the size of~$q$, to compute the attack graph of 
 $q(\vec{z})$~\cite{DBLP:journals/tods/Wijsen12}, test whether condition~\ref{it:noCyclicOrStrongAttacks} in Definition~\ref{def:newclass} is satisfied, and construct~$V$ as in the statement of Proposition~\ref{pro:defineX}.
 By adding to $q(\vec{z})$, for every variable $x$, a fresh atom~$P_{x}(\underline{x})$, the attack graph allows determining the set $n_{x}$ of all atoms not attacking~$x$.  
 Membership of $x$ in $\frozen{q}$ is determined by testing whether $\fdset{n_{x}}\models\fd{\emptyset}{x}$, which is in linear time in the size of $n_{x}$. 
If condition~\ref{it:noCyclicOrStrongAttacks} is not satisfied, answer ``no''; otherwise proceed as follows.
For every atom~$R$ in $\body{q(\vec{z})}$, test whether, in $q(\vec{z})$'s query graph, some variable of $V$ is reachable from some variable in $\notkey{R}$ without using variables from $\key{R}\cup\frozen{q}$.
This can be done in quadratic time.
If any such reachability test succeeds, then it is correct to answer ``no'', because $q(\vec{z})$ has no id-set; otherwise answer~``yes''.
The correctness follows from Proposition~\ref{pro:defineX}: if $q(\vec{z})$ has an id-set, then $V$ is an id-set.  
Therefore, if $V$ falsifies condition~\ref{it:separation} in Definition~\ref{def:newclass}, then $V$ is not an id-set, and hence $q(\vec{z})$ has no id-set.
Finally, we note that, by construction, $V$ will satisfy condition~\ref{it:source} in Definition~\ref{def:newclass}.
\end{proof}

\section{Proofs for Section \ref{sec:soundness}}


\begin{proof}[Proof of Lemma~\ref{lem:onlyOneVal}]
Assume $\db \models q(\vec{a}, \vec{b}, \vec{c}_{1})$ and $\db \models q(\vec{a}, \vec{b}, \vec{c}_{2})$.

Let $M$ be a component in the attack graph of $q$. 
By Definition~\ref{def:newclass}, there is an unattacked atom $R$ in $M$ such that $\fdset{q}\models\fd{\vec{x}}{\key{R}}$. 
By Lemma~\ref{lem:notAttackedKeyDetermineComponent}, for every variable $v$ occurring in $M$, $\fdset{q}\models\fd{\key{R}}{v}$.
It follows that for every variable $v$ occurring in $M$, $\fdset{q} \models \fd{\vec{z}\cdot\vec{x}}{v}$.
Since $\vec{z}$ and $\vec{x}$ do not depend on~$M$, it follows  $\fdset{q} \models \fd{\vec{z}\cdot\vec{x}}{\vec{w}}$.
Since $\db$ is a consistent database instance, from $\fdset{q} \models \fd{\vec{z}\cdot\vec{x}}{\vec{w}}$, $\db \models q(\vec{a}, \vec{b}, \vec{c}_{1})$, and $\db \models q(\vec{a}, \vec{b}, \vec{c}_{2})$, it follows that $\vec{c}_{1} = \vec{c}_{2}$ by Lemma~4.3 in~\cite{DBLP:journals/tods/Wijsen12}.
\end{proof}


\begin{lemma}\label{lem:valsForOpti}
Let $q(\vec{z})$ be self-join free conjunctive query in $\cnewclass$, and let $\vec{x}$ be an id-set for it.
Let $q'(\vec{z}, \vec{x})$ be the query $\makefree{q}{\vec{x}}$.
Let $\db$ be a database instance, and $\vec{c}$ a tuple of constants, of size $\arity{\vec{z}}$, such that $\db\cqamodels q(\vec{c})$.
Let $\tuple{\vec{d}_{1}, \ldots, \vec{d}_{n}}$ be a sequence that contains all (and only) tuples of constants $\vec{d}_{i}$, of size $\arity{\vec{x}}$, such that $\db \models q'(\vec{c}, \vec{d}_{i})$.
There is a sequence $\tuple{\theta_{1},\ldots,\theta_{n}}$ of valuations over $\vars{q}$ such that for every $i,j\in\{1,\dots,n\}$, for every atom $R \in \body{q}$:
\begin{enumerate}[label=(\alph*)]
    \item\label{it:valOptiCond1} $\theta_{i}(R) \in \db$; 
    \item\label{it:valOptiCond2} $\theta_{i}(\vec{z}) = \vec{c}$ and $\theta_{i}(\vec{x}) = \vec{d}_{i}$; and
    \item\label{it:valOptiCond3} 
    if $\theta_{i}(R)$ and $\theta_{j}(R)$ are key-equal, then they are equal.
\end{enumerate}
\end{lemma}
\begin{proof}
Since $\db\models q'(\vec{c},\vec{d_{i}})$ for every $i\in\{1,\ldots,n\}$, there is a sequence $\tuple{\theta_{1}, \ldots, \theta_{n}}$ of valuations over $\vars{q}$ such that for every~$i\in\{1,\ldots,n\}$ and every atom~$R\in\body{q}$, conditions~\ref{it:valOptiCond1} and \ref{it:valOptiCond2} are satisfied.
Assume that condition~\ref{it:valOptiCond3} is violated. Then there are $i,j\in\{1,\ldots,n\}$ such that for some atom $R\in\body{q}$, $\theta_{i}(R)$ and $\theta_{j}(R)$ are key-equal but distinct. 
Let $M$ be the smallest subset of $\body{q}$ such that:
\begin{itemize}
    \item for every atom $R \in\body{q}$, if $\theta_{i}(R)$ and $\theta_{j}(R)$ are key-equal but distinct, then $R \in M$ (hence $M\neq\emptyset$); and
    \item 
          \emph{Closure Property:} if $M$ contains $S$ and $\theta_{i}(v)\neq\theta_{j}(v)$ for some $v\in\vars{S}$, then $M$ contains every atom of $q$ in which $v$ occurs.
\end{itemize}

We show that $\theta_{i}$ and $\theta_{j}$ agree on every variable of~$\vec{x}$ that occurs in~$M$.
Assume for the sake of contradiction that some variable $x$ of $\vec{x}$ satisfies $\theta_{i}(x)\neq\theta_{j}(x)$ and occurs in some atom~$M$. 
By construction of $M$, there is an atom $R_{0}$ in $\body{q}$, and a sequence of variables $\tuple{v_{1},\ldots,v_{n}}$, with $n\geq 1$, such that:
\begin{itemize}
    \item $v_{n}=x$;
    \item for every $k\in\{1,\ldots,n\}$, $\theta_{i}(v_{k})\neq\theta_{j}(v_{k})$;
    \item $\theta_{i}(R_{0})$ and $\theta_{j}(R_{0})$ are key-equal but distinct, and $v_{1}\in\notkey{R_{0}}$;
    \item every two adjacent variables in the sequence occur together in some atom of~$M$.
\end{itemize}
Note that the sequence may reduce to $\tuple{v_{1}}$ with $v_{1}=x$.
Then $\tuple{v_{1},\ldots,v_{n}}$ is a path in the query graph of~$q$ that uses no variable of $\key{R_0}$.

We show that no variable among $v_{1},\ldots,v_{n}$ is frozen.
Assume for the sake of contradiction that some $v_{k}$ is frozen ($1\leq k\leq n$).
Let $q^{*}(\vec{z},v_{k})\defeq\makefree{q}{v_{k}}$.
Since $\db\cqamodels q(\vec{c})$, it follows from (the proof of) \cite[Lemma~11]{DBLP:journals/mst/KoutrisW21} that for all constants $f_{1},f_{2},$ if $\db\models q^{*}(\vec{c},f_{1})$ and $\db\models q^{*}(\vec{c},f_{2})$, then $f_{1}=f_{2}$.
Since $\db\models q^{*}(\vec{c},\theta_{i}(v_{k}))$ and $\db\models q^{*}(\vec{c},\theta_{j}(v_{k}))$, it follows $\theta_{i}(v_{k})=\theta_{j}(v_{k})$, a contradiction.
We conclude by contradiction that no variable among $v_{1},\ldots,v_{n}$ is frozen.

Then $\tuple{v_{1},\ldots,v_{n}}$ is a path in the query graph of~$q$ between a variable of $\notkey{R_{0}}$ and a variable of~$\vec{x}$ that uses no variable of $\key{R_0}\cup\frozen{q}$.
Consequently, $\vec{x}$ violates condition~\ref{it:separation} in Definition~\ref{def:newclass}.
Then $\vec{x}$ is not an id-set, a contradiction.
We conclude by contradiction that $\theta_{i}$ and $\theta_{j}$ agree on every variable of $\vec{x}$ that occurs in~$M$.

Assume $i<j$. Let $\theta_{j}^{*}$ be a valuation over $\vars{q}$ such that for every $v\in\vars{q}$,
$$\theta_{j}^{*}(v) =  
\begin{cases}
	\theta_{i}(v) & \textnormal{if $x$ occurs in some atom of $M$;}\\
	\theta_{j}(v) & \textnormal{otherwise.}
\end{cases}$$
Since $\theta_{i}$ and $\theta_{j}$ agree on every variable of $\vec{x}$ that occurs in~$M$, it follows $\theta_{j}^{*}(\vec{x})=\theta_{j}(\vec{x})$, and thus $\theta_{j}^{*}(\vec{x})=\vec{d}_{j}$.
We now argue that $\theta_{j}^{*}(R)\in\db$ for every $R\in\body{q}$. We consider two cases:
\begin{description}
\item[Case $R\in M$.]
Then $\theta_{j}^{*}(R)=\theta_{i}(R)$, and we have that $\theta_{i}(R)\in\db$;
\item[Case $R\notin M$.]
Then $\theta_{j}^{*}(R)=\theta_{j}(R)$,  and we have that $\theta_{j}(R)\in\db$.
Note here that if some variable $u$ of $\vars{R}$ occurs in~$M$, then $R\not\in M$ implies $\theta_{i}(u)=\theta_{j}(u)$.
\end{description}

The sequence $\tuple{\theta_{1},\ldots,\theta_{i},\ldots,\theta_{j-1},\theta_{j}^{*},\theta_{j+1},\ldots,\theta_{n}}$ therefore still satisfies conditions~\ref{it:valOptiCond1} and \ref{it:valOptiCond2}.
By our construction, for every atom $R\in\body{q}$, if the facts $\theta_{j}^{*}(R)$ and $\theta_{i}(R)$ are key-equal, then they are equal. 
By repeatedly removing primary-key violations in this way, we eventually obtain a sequence that also satisfies~\ref{it:valOptiCond3}.
Note that our procedure will not be trapped in an infinite loop, because whenever two distinct key-equal facts have to be made equal, say $\theta_{i}(R)$ and $\theta_{j}(R)$, then only the valuation with the larger index will be modified.
This means that $\theta_{1}$ will never be modified; $\theta_{2}$ will only be modified with respect to $\theta_{1}$; and so on.
\end{proof}

We can now give the proof of Lemma~\ref{lem:hasOptiRepair}.

\begin{proof}[Proof of Lemma~\ref{lem:hasOptiRepair}]
Let $\tuple{\vec{d}_{1}, \ldots, \vec{d}_{n}}$ be a sequence that contains all (and only) tuples of constants $\vec{d}_{i}$, of size $\arity{\vec{x}}$, such that $\db \models q'(\vec{c}, \vec{d}_{i})$.

By Lemma~\ref{lem:valsForOpti}, there is a sequence $\tuple{\theta_{1}, \ldots, \theta_{n}}$ of valuations over $\vars{q}$ such that for every $i,j\in\{1,\dots,n\}$, for every atom $R \in \body{q}$:
\begin{enumerate}[label=(\alph*)]
    \item $\theta_{i}(R) \in \db$; 
    \item $\theta_{i}(\vec{z}) = \vec{c}$ and $\theta_{i}(\vec{x}) = \vec{d}_{i}$; and
    \item if $\theta_{i}(R)$ and $\theta_{j}(R)$ are key-equal, then $\theta_{i}(R)=\theta_{j}(R)$.
\end{enumerate}
From the last property, it follows that there is a repair $\rep$ of $\db$ such that for every $i\in\{1,\ldots,n\}$, for every atom $R \in\body{q}$, we have $\theta_{i}(R)\in\rep$. 
Since $\rep\models q'(\vec{c},\vec{d}_{i})$ for every $i\in\{1,\ldots,n\}$, the repair $\rep$ is an optimistic repair of $\db$ with respect to $\substitute{q'(\vec{z}, \vec{x})}{\vec{z}}{\vec{c}}$.
\end{proof}

Finally, we show Lemma~\ref{lem:hasPessRepair}.

\begin{proof}[Proof of Lemma~\ref{lem:hasPessRepair}]
From~\cite{DBLP:journals/tods/KoutrisW17}, it follows that whenever $\vec{x}$ is a sequence of unattacked variables of a self-join-free Boolean conjunctive query,
then every database instance $\db$ has a repair $\rep$ such that for every sequence $\vec{d}$ of constants, of arity~$\arity{\vec{x}}$,
$\rep\models\substitute{q}{\vec{x}}{\vec{d}}$ implies $\db\cqamodels\substitute{q}{\vec{x}}{\vec{d}}$.
The proof is now straightforward by using that all variables of an id-set are unattacked (by item~\ref{it:idsetunattacked} in Lemma~\ref{lem:xIsKeyofNA}), and that free variables can be treated as constants.
\end{proof}

\section{Proofs for Section \ref{sec:completeness}}

\begin{lemma}\label{lem:notparsimonious}
Let $q^{*}(\vec{z},\vec{x},\vec{w})$ be a full self-join-free conjunctive query. 
Let $q'(\vec{z},\vec{x})=\makebound{q^{*}}{\vec{w}}$ and $q(\vec{z})=\makebound{q'}{\vec{x}}$.
If there exists a consistent database instance $\db$ such that 
$\db\models q(\vec{a},\vec{b},\vec{c_{1}})$ and $\db\models q(\vec{a},\vec{b},\vec{c_{2}})$ with $c_{1}\neq c_{2}$,
then $q(\vec{z})$ does not admit parsimonious counting on~$\vec{x}$.
Here, $\vec{a}$, $\vec{b}$, $\vec{c}_{1}$, and $\vec{c}_{2}$ are understood to be tuples of constants, where $\vec{a}$ is of arity $\arity{\vec{z}}$, $\vec{b}$ of arity $\arity{\vec{y}}$,
and $\vec{c_{i}}$ of arity $\arity{\vec{z}}$ for $i\in\{1,2\}$.
\end{lemma}
\begin{proof}
Let $\db$ be a consistent database instance such that
$\db\models q(\vec{a},\vec{b},\vec{c_{1}})$ and $\db\models q(\vec{a},\vec{b},\vec{c_{2}})$ with $c_{1}\neq c_{2}$.
Let $D$ be the active domain of~$\db$.
Define:
\begin{eqnarray*}
\calD & \defeq & \{\vec{b}\in D^{\arity{\vec{x}}}\mid\db\models q'(\vec{a},\vec{b})\},\\
\calD^{*} & \defeq & \{\vec{b}\cdot\vec{c}\in D^{\arity{\vec{x}}}\cdot D^{\card{\vec{w}}}\mid\db\models q^{*}(\vec{a},\vec{b},\vec{c})\}.
\end{eqnarray*}
Clearly, for every $\vec{b}\in\calD$, there exists $\vec{c}\in D^{\arity{\vec{w}}}$ such that $\vec{b}\cdot\vec{c}\in\calD^{*}$.
Therefore, $\card{\calD}\leq\card{\calD^{*}}$.
Moreover, from $\db\models q^{*}(\vec{a},\vec{b},\vec{c}_{1})$ and $\db\models q^{*}(\vec{a},\vec{b},\vec{c}_{2})$, it follows $\card{\calD}<\card{\calD^{*}}$.
The query $\mycount{q^{*}}{\vec{z}}$ on $\db$ returns $(\vec{a},\card{\calD^{*}})$.
Since $\db$ is consistent, the query $\cqacount{q^{*}}{\vec{z}}$ on $\db$ returns $(\vec{a},[\card{\calD^{*}},\card{\calD^{*}}])$.
From $\card{\calD}\neq\card{\calD^{*}}$, it is correct to conclude that $q(\vec{z})$ does not admit parsimonious counting on~$\vec{x}$.
\end{proof}

\begin{proof}[Proof of Lemma~\ref{lem:strongImpliesMultipleValuations}]
The proof is by contraposition.
Assume that the attack graph of  $q'(\vec{z}, \vec{x})$ has a strong attack.
We can assume atoms $R$ and $S$ such that the attack graph of $q'(\vec{z},\vec{x})$ has a strong attack from $R$ to~$S$.
Let $\theta,\mu$ be two valuations over $\vars{q}$ such that for every variable $u\in\vars{q}$,
\begin{equation}\label{eq:thetamu}
\textnormal{$\theta(u)=\mu(u)$ if and only if $\fdset{q'}\models\fd{\key{R}}{u}$.}   
\end{equation}
Since $\fdset{q'}$ contains $\fd{\emptyset}{\free{q}}$ by definition, we have $\fdset{q'}\models\fd{\key{R}}{\key{R}\cup\free{q}}$ by Armstrong's augmentation axiom.
Consequently, if $\fdset{q'}\models\fd{\key{R}\cup\free{q}}{u}$, then $\fdset{q'}\models\fd{\key{R}}{u}$ (the converse is obvious).
Thus, $\fdset{q'}\models\fd{\key{R}}{u}$ if and only if $\fdset{q'}\models\fd{\key{R}\cup\free{q}}{u}$.

Let $\db=\theta(\body{q})\cup\mu(\body{q})$.
We show that $\db$ is a consistent database instance, i.e., whenever two facts in $\db$ are key-equal, then they are equal.
Since $q'$ is self-join-free, it suffices to consider any atom $F$ in $\body{q}$ such that for every $u\in\key{F}$, $\theta(u)=\mu(u)$.
By Eq.~\eqref{eq:thetamu}, it follows $\fdset{q'}\models\fd{\key{R}}{\key{F}}$.
Since $\fdset{q'}$ contains $\fd{\key{F}}{\vars{F}}$, it follows $\fdset{q'}\models\fd{\key{R}}{\vars{F}}$ by Armstrong's transitivity axiom.
By Eq.~\eqref{eq:thetamu}, $\theta(F)=\mu(F)$.

Let $\vec{w}$ be a tuple of the bound variables of $q'(\vec{z},\vec{x})$, and let $q^{*}(\vec{z},\vec{x},\vec{w})=\makefree{q'}{\vec{w}}$.
Since $\fdset{q'}\models\fd{\key{R}}{\free{q}}$, we have $\theta(\vec{z})=\mu(\vec{z})$ and $\theta(\vec{x})=\mu(\vec{x})$ by Eq.~\eqref{eq:thetamu}.
Since the attack from $R$ to~$S$ is strong, there exists $v\in\key{S}$ such that $\fdset{q'}\not\models\fd{\key{R}}{v}$,
and hence $\theta(v)\neq\mu(v)$ by Eq.~\eqref{eq:thetamu}.
Since $v$ is necessarily in $\vec{w}$, we have $\theta(\vec{w})\neq\mu(\vec{w})$.
Let $\vec{a}=\theta(\vec{z})=\mu(\vec{z})$, $\vec{b}=\theta(\vec{x})=\mu(\vec{x})$, $\vec{c_{1}}=\theta(\vec{w})$, and $\vec{c_{2}}=\mu(\vec{w})$.
Thus, $\vec{c}_{1}\neq\vec{c}_{2}$.
Clearly, $\db\models q^{*}(\vec{a},\vec{b},\vec{c}_{1})$ and $\db\models q^{*}(\vec{a},\vec{b},\vec{c}_{2})$.
By Lemma~\ref{lem:notparsimonious}, it is correct to conclude that $q(\vec{z})$ does not admit parsimonious counting on~$\vec{x}$.
%
\end{proof}

\begin{proof}[Proof of Lemma~\ref{lem:freeBadComponentImpliesMultipleValuations}]
The proof is by contraposition.
Assume that $\vec{x}$ does not satisfy condition~\ref{it:source} in Definition~\ref{def:newclass}.
We can assume that the attack graph of $q(\vec{z})$ contains an unattacked atom~$S$ such that $\fdset{q}\not\models\fd{\vec{x}}{\key{S}}$.  
Let $\theta,\mu$ be two valuations over $\vars{q}$ such that for every variable $u\in\vars{q}$,
\begin{equation}\label{eq:thetamubis}
\textnormal{$\theta(u)=\mu(u)$ if and only if $\fdset{q}\models\fd{\vec{x}}{u}$.}   
\end{equation}
Let $\db=\theta(\body{q})\cup\mu(\body{q})$.
It can be seen that $\db$ is consistent; the reasoning is similar to that in the proof of Lemma~\ref{lem:strongImpliesMultipleValuations}.

Let $\vec{w}$ be a tuple of the bound variables of $q'(\vec{z},\vec{x})$, and let $q^{*}(\vec{z},\vec{x},\vec{w})=\makefree{q'}{\vec{w}}$.
Since $\fdset{q}\models\fd{\vec{x}}{\vec{x}\cdot\vec{z}}$, we have $\theta(\vec{z})=\mu(\vec{z})$ and $\theta(\vec{x})=\mu(\vec{x})$ by Eq.~\eqref{eq:thetamubis}.
Since $\fdset{q}\not\models\fd{\vec{x}}{\key{S}}$, there exists $v\in\key{S}$ such that $\fdset{q}\not\models\fd{\vec{x}}{v}$,
and hence $\theta(v)\neq\mu(v)$ by Eq.~\eqref{eq:thetamubis}.
Since $v$ is necessarily in $\vec{w}$, we have $\theta(\vec{w})\neq\mu(\vec{w})$.
Let $\vec{a}=\theta(\vec{z})=\mu(\vec{z})$, $\vec{b}=\theta(\vec{x})=\mu(\vec{x})$, $\vec{c_{1}}=\theta(\vec{w})$, and $\vec{c_{2}}=\mu(\vec{w})$.
Thus, $\vec{c}_{1}\neq\vec{c}_{2}$.
Clearly, $\db\models q^{*}(\vec{a},\vec{b},\vec{c}_{1})$ and $\db\models q^{*}(\vec{a},\vec{b},\vec{c}_{2})$.
By Lemma~\ref{lem:notparsimonious}, it is correct to conclude that $q(\vec{z})$ does not admit parsimonious counting on~$\vec{x}$.
\end{proof}

\begin{proof}[Proof of Lemma~\ref{lem:freeBadPathImpliesNoOptiRepair}]
We can assume an atom $R$ in $q(\vec{z})$ such that the query graph contains a path $\pi=\tuple{v_{1},v_{2},\ldots,v_{n}}$ from some variable $v_{1}$ in $\notkey{R}$ to a variable $v_{n}$ in $\vec{x}$ such that $v_{1},\ldots,v_{n}\notin\key{R}\cup\frozen{q}$.
Let $U=\{u\in\vars{q}\mid \fdset{q}\models\fd{\emptyset}{u}\}$.
We distinguish two cases.
\begin{description}
\item[Case that $v_{1},\ldots,v_{n}\notin U$.]
Let $\theta,\mu,\gamma$ be three valuations over $\vars{q}$ such that for every variable $u\in\vars{q}$,
\begin{eqnarray}
\theta(u)\neq\mu(u) & \textnormal{if and only if} & u\in\{v_{1},\dots,v_{n}\};\label{eq:valthree}\\
\gamma(u)=\theta(u) & \textnormal{if and only if} & u\in U;\label{eq:valone1}\\
\gamma(u)=\mu(u) & \textnormal{if and only if} & u\in U.\label{eq:valone2}
\end{eqnarray}
We have that $\vars{\vec{z}}\subseteq U$.
The valuations are well-defined because $v_{1},\ldots,v_{n}\notin U$.
Moreover, we can choose $\gamma(\vec{z})=\vec{c}$, as follows:
$$
\begin{array}{|*{10}{c}c}
\multicolumn{3}{c}{} & \multicolumn{4}{c}{\overbrace{\rule{80pt}{0pt}}^{\textnormal{variables of $U$}}} & \multicolumn{3}{c}{\overbrace{\rule{70pt}{0pt}}^{\textnormal{other variables}}}\\[-3pt]
v_{1} & \dotsm & v_{n} & \vec{z} & v_{n+1} & \dotsm & v_{m} & v_{m+1} & \dotsm & v_{\ell}\\\cline{1-10}
0 & \dotsm & 0 & \vec{c} & 0 & \dotsm & 0 & 0 & \dotsm & 0 & (\theta)\\
1 & \dotsm & 1 & \vec{c} & 0 & \dotsm & 0 & 0 & \dotsm & 0 & (\mu)\\
2 & \dotsm & 2 & \vec{c} & 0 & \dotsm & 0 & 2 & \dotsm & 2 & (\gamma)
\end{array}
$$

Let $\db=\theta(\body{q})\cup\mu(\body{q})\cup\gamma(\body{q})$.
We show that every repair of $\db$ includes $\gamma(\body{q})$, and hence $\db\cqamodels q(\vec{c})$.
To this end, let $F$ be an atom of~$q$.
If $\key{F}\nsubseteq U$, then by Eq.~\eqref{eq:valone1} and Eq.~\eqref{eq:valone2}, $\gamma(F)$ is key-equal to neither $\theta(F)$ nor $\mu(R)$, and, consequently, every repair of $\db$ will contain~$\gamma(F)$.
Assume next that $\key{F}\subseteq U$.
Since $\fdset{q}$ contains $\fd{\key{F}}{\vars{F}}$, it follows that  $\vars{F}\subseteq U$.
Then, $\gamma(F)=\theta(F)=\mu(F)$ by~Eq.~\eqref{eq:valone1} and Eq.~\eqref{eq:valone2}.

Note that $\theta(R)$ and $\mu(R)$ are key-equal, because the path $\pi$ contains no variable of~$\key{R}$.
Also $\theta(R)\neq\mu(R)$, because $v_{1}\in\notkey{R}$ and $\theta(v_{1})\neq\mu(v_{1})$.
It follows that every repair must contain either $\theta(R)$ or $\mu(R)$, but not both.

Since $\theta(v_{n})\neq\mu(v_{n})$, we have $\theta(\vec{x})\neq\mu(\vec{x})$.
Let $\vec{b}_{\theta}=\theta(\vec{x})$ and $\vec{b}_{\mu}=\mu(\vec{x})$.
Clearly, $\db\models q'(\vec{c},\vec{b}_{\theta})$ and $\db\models q'(\vec{c},\vec{b}_{\mu})$.

Let $\rep$ be a repair that contains $\mu(R)$, and hence $\theta(R)\notin\rep$. We show $\rep\not\models q'(\vec{c},\vec{b_{\theta}})$. 
Assume for the sake of contradiction that there is a valuation $\nu$ over $\vars{q}$ such that $\nu(\body{q})\subseteq\rep$ and $\nu(\vec{x})=\vec{b}_{\theta}$. 
We have that $\nu$ maps the sequence $\pi$ to $\tuple{\nu(v_{1}),\dots,\nu(v_{n})}$.
We have $\nu(v_{n})=\theta(v_{n})$.
Since every two adjacent variables occur together in some atom,
it follows from Eq.~\eqref{eq:valthree} that $\nu(v_{1})=\theta(v_{1})$, hence $\nu(R)=\theta(R)$, a contradiction. 
We conclude by contradiction that $\rep\not\models q'(\vec{c},\vec{b_{\theta}})$.
By similar reasoning, for a repair $\rep$ that contains $\theta(R)$, we have $\rep\not\models q'(\vec{c},\vec{b_{\mu}})$.
It follows that for every repair $\rep$ of $\db$, either $\rep\not\models q'(\vec{c},\vec{b_{\theta}})$ or $\rep\not\models q'(\vec{c},\vec{b_{\mu}})$.

\item[Case that $\{v_{1},\ldots,v_{n}\}\cap U\neq\emptyset$.]
We can assume a greatest integer $\ell\in\{1,\dots,n\}$ such that $v_{\ell}\in U$.
That is, $v_{\ell+1}, v_{\ell+2}, \dots, v_{n}\notin U$.
We can assume a sequence $\sigma$ of atoms that is a sequential proof of $\fdset{q}\models\fd{\emptyset}{v_{\ell}}$.
Since~$v_{\ell}$ is not frozen, there is some atom $S$ in the sequential proof~$\sigma$ such that $\attacks{S}{v_{\ell}}{q}$, and hence $v_{\ell}\notin\keycl{S}{q}$.

We show $v_{\ell+1},\ldots,v_{n}\notin\keycl{S}{q}$.
Assume for the sake of contradiction $j\in\{\ell+1,\ldots,n\}$ such that $\fdset{q\setminus\{S\}}\models\fd{\key{S}}{v_{j}}$.
Since $\fdset{q\setminus\{S\}}\models\fd{\emptyset}{\key{S}}$ (because $S$ occurs in the sequential proof $\sigma$),
we obtain $\fdset{q\setminus\{S\}}\models\fd{\emptyset}{v_{j}}$, hence $v_{j}\in U$, a contradiction.
Consequently, $\ell\leq n$ and
$v_{\ell},v_{\ell+1}, \dots, v_{n}\notin\keycl{S}{q}$.

Let $\theta$, $\mu$ be two valuations over $\vars{q}$ such that for every $u\in\vars{q}$,
$\theta(u)=\mu(u)$ if and only if $u\in\keycl{S}{q}$.
Since $\vars{\vec{z}}\subseteq\keycl{S}{q}$ by definition (even though some variables of~$\vec{z}$ may not occur in $q\setminus\{S\}$), we can choose  $\theta(\vec{z})=\mu(\vec{z})=\vec{c}$.

Let $\db=\theta(q)\cup\mu(q)$.
The only facts in $\db$ that are key-equal but distinct are~$\theta(S)$ and~$\mu(S)$.
Note here that since $S$ attacks some variable, there is a variable of $\notkey{S}$ that is not in~$\keycl{S}{q}$.
Consequently, $\db$ has exactly two repairs, denoted $\rep_{1}\defeq\db\setminus\{\mu(S)\}$ and $\rep_{2}\defeq\db\setminus\{\theta(S)\}$.

Let $\vec{d}_{1}=\theta(\vec{x})$ and $\vec{d}_{2}=\mu(\vec{x})$.
Since $v_{n}\notin\keycl{S}{q}$, we have $\theta(v_{n})\neq\mu(v_{n})$. 
Since $v_{n}\in\vars{\vec{x}}$, $\vec{d}_{1}\neq\vec{d}_{2}$.
For the query $q'(\vec{z},\vec{x})\defeq\makefree{q}{\vec{x}}$,  we obviously have $\db\models q'(\vec{c},\vec{d}_{1})$ and $\db\models q'(\vec{c},\vec{d}_{2})$.
Using the same proof as~\cite[Proposition~6.4]{DBLP:journals/tods/Wijsen12},
we have $\rep_{1}\models q'(\vec{c},\vec{d}_{1})$ and $\rep_{1}\not\models q'(\vec{c},\vec{d}_{2})$.
Symmetrically, $\rep_{2}\not\models q'(\vec{c},\vec{d}_{1})$ and $\rep_{2}\models q'(\vec{c},\vec{d}_{2})$.
Clearly, $\db\cqamodels q(\vec{c})$.
\end{description}
In both cases, there is a database $\db$ such that $\db\cqamodels q(\vec{c})$, and $\db$ has no optimistic repair with respect to $\substitute{q'}{\vec{z}}{\vec{c}}$.
\end{proof}

\begin{lemma}\label{lem:pathInSequentialProofs}
Let $q(\vec{z})$ be a self-join-free conjunctive query.
Let $\vec{x}$ be a sequence of distinct bound variables of $q(\vec{z})$.
Let $q'(\vec{z},\vec{x})=\makefree{q}{\vec{x}}$.
Let $Y\cup\{w\}$ be a set of bound variables of~$q(\vec{z})$.
If $\fdset{q}\not\models\fd{Y}{w}$ and $\fdset{q'}\models\fd{Y}{w}$, then there is a path in the query graph of $q$ between some variable of $\vec{x}$ and $w$ such that for every variable~$u$ on the path, $\fdset{q}\not\models\fd{Y}{u}$.
\end{lemma}
\begin{proof}
Assume $\fdset{q}\not\models\fd{Y}{w}$ and $\fdset{q'}\models\fd{Y}{w}$.
It is easily verified that $\fdset{q'}\equiv\fdset{q}\cup\{\fd{\emptyset}{\vec{x}}\}$.
From $\fdset{q'}\models\fd{Y}{w}$, it follows $\fdset{q}\models\fd{Y\cup\vec{x}}{w}$. 
Let $U=\{u\in\vars{q}\mid\fdset{q}\models\fd{Y}{u}\}$.
Clearly, $Y\cup\free{q}\subseteq U$ and $w\not\in U$.
If $w\in\vars{\vec{x}}$, then the empty path between $w$ and $w$ is a path between some variable of~$\vec{x}$ and~$w$ that uses no variable of~$U$, and the desired result obtains.
From here on, assume $w\notin\vars{\vec{x}}$.
We can assume a sequence of atoms  $\tuple{F_{1}, \ldots, F_{n}}$ that is a minimal sequential proof of $\fdset{q} \models \fd{Y\cup\vec{x}}{w}$. 
Note that $w \in \notkey{F_{n}}$. 
Let $M$ be a minimal subset of $\{F_{1}, \ldots, F_{n}\}$ such that:
\begin{itemize}
    \item $F_{n} \in M$; and
    \item for all $i,j\in\{1, \ldots, n-1\}$ such that $i<j$, if $F_{j}\in M$ and some variable of $\key{F_{j}}\setminus U$ occurs in $F_{i}$, then $F_{i}\in M$.
\end{itemize}

By construction of $M$, if some variable of $\vec{x}$ occurs in $M$, 
then there is a path in the query graph between a variable of~$\vec{x}$ and $w$ that uses no variable in~$U$, and hence the desired result obtains.
Therefore, to conclude the proof, it suffices to show that some variable of $\vec{x}$ occurs in~$M$.

Assume for the sake of contradiction that no variable of $\vec{x}$ occurs in $M$.
By construction of~$M$, it holds that for every atom $F_{j}\in M$, for every variable $v$ in $\key{F_{j}}$, either $v\in U$, or $v$ appears in some other atom $F_{i}\in M$ with $i<j$.
Consequently, if we restrict $\tuple{F_{1},\ldots,F_{n}}$ to the atoms of~$M$, we obtain a sequential proof of $\fdset{q}\models\fd{U}{w}$.
Since $\fdset{q}\models\fd{Y}{U}$ is obvious, it is correct to conclude $\fdset{q}\models\fd{Y}{w}$, a contradiction.
\end{proof}


\begin{lemma}\label{lem:attacksKept}
Let $q(\vec{z})$ be a self-join-free conjunctive query.
Let $\vec{x}$ be a sequence of distinct bound variables of $q(\vec{z})$.
Let $q'(\vec{z},\vec{x})=\makefree{q}{\vec{x}}$.
Assume that $\vec{x}$ satisfies condition~\ref{it:separation} in Definition~\ref{def:newclass} of $\cnewclass$.
For all atoms $R$ and $S$ of $q$,
if $\attacks{R}{S}{q}$, then $\attacks{R}{S}{q'}$.
\end{lemma}
\begin{proof}
Assume that $\attacks{R}{S}{q}$.
We can assume a witness $\tuple{v_{1}, \ldots, v_{n}}$ that $R$ attacks $S$ in $q(\vec{z})$.
Clearly, for every $i\in\{1,\dots,n\}$, we have $v_{i}\notin\key{R}$, and, by Lemma~\ref{lem:witnessNotFrozen}, $v_{i}\notin\frozen{q}$. 
Assume for the sake of contradiction that $\notattacks{R}{S}{q'}$.
It suffices to show that $\vec{x}$ does not satisfy condition~\ref{it:separation} in Definition~\ref{def:newclass}, which is a contradiction with the lemma's hypothesis.

Since $\tuple{v_{1}, \ldots, v_{n}}$ is not a witness that $R$ attacks $S$ in $q'(\vec{z},\vec{x})$,
we can assume some $i\in\{1, \ldots, n\}$ such that $v_{i}\notin\keycl{R}{q}$ and $v_{i}\in\keycl{R}{q'}$. 
Consequently, $\fdset{q\setminus\{R\}}\not\models\fd{\key{R}}{v_{i}}$ and $\fdset{q'\setminus\{R\}}\models\fd{\key{R}}{v_{i}}$.
By Lemma~\ref{lem:pathInSequentialProofs}, we can assume that the query graph of $q\setminus\{R\}$ has a path $\pi=\tuple{x,\ldots,v_{i}}$ between some variable~$x$ in $\vars{\vec{x}}\cap\vars{q'\setminus\{R\}}$ and~$v_{i}$ such that for every variable~$u\in\vars{\pi}$, we have $\fdset{q\setminus\{R\}}\not\models\fd{\key{R}}{u}$.
Since $\vars{\pi}\cap\keycl{R}{q}=\emptyset$ and $\attacks{R}{v_{i}}{q}$, it follows that $\attacks{R}{u}{q}$ for every $u\in\vars{\pi}$.
From Lemma~\ref{lem:witnessNotFrozen}, it follows $\vars{\pi}\cap\frozen{q}=\emptyset$.

Then, the path $\pi\cdot\tuple{v_{i-1}, v_{i-2},\ldots v_{1}}$ is a path in the query graph of~$q$ between a variable of~$\vec{x}$ and a variable of $\notkey{R}$ that uses no variable of $\key{R}\cup\frozen{q}$.
It follows that~$\vec{x}$ violates condition~\ref{it:separation} in Definition~\ref{def:newclass} of $\cnewclass$.
\end{proof}

The statement of Lemma~\ref{lem:strongImpliesStrong} is almost identical to Lemma~\ref{lem:attacksKept};
the only difference is that it deals with strong attacks.

\begin{proof}[Proof of Lemma~\ref{lem:strongImpliesStrong}]
Assume that $\vec{x}$ satisfies condition~\ref{it:separation} in Definition~\ref{def:newclass}.
The proof is by contraposition:
we assume that the attack graph of $q'(\vec{z},\vec{x})$ has no strong attack from $R$ to $S$;
we need to show that the attack graph of $q(\vec{z})$ has no strong attack from $R$ to $S$.

If $\notattacks{R}{S}{q'}$, then $\notattacks{R}{S}{q}$ by Lemma~\ref{lem:attacksKept}, and the desired result obtains.
Assume from here on that $\attacks{R}{S}{q'}$, which must be a weak attack.
Assume for the sake of contradiction that $q(\vec{z})$'s attack graph has a strong attack $\attacks{R}{S}{q}$.

Consequently, $\fdset{q}\not\models\fd{\key{R}}{\key{S}}$ and $\fdset{q'}\models \fd{\key{R}}{\key{S}}$.
We can assume a variable $y\in\key{S}$ such that  $\fdset{q}\not\models\fd{\key{R}}{y}$ and $\fdset{q'}\models\fd{\key{R}}{y}$.
By Lemma~\ref{lem:pathInSequentialProofs}, for some variable $x$ in~$\vec{x}$, there is a path $\pi=\tuple{x,\ldots,y}$ in the query graph of~$q$ such that for every $u\in\vars{\pi}$, $\fdset{q}\not\models\fd{\key{R}}{u}$.
Consequently, for every $u\in\vars{\pi}$, $u\not\in\key{R}\cup\frozen{q}$.

There is a witness $\tuple{v_{1},\ldots,v_{n}}$ of $\attacks{R}{S}{q}$ with $v_{n}\in\key{S}$.  It will be the case that $v_{1}\in\notkey{R}$ and $v_{1},\ldots,v_{n}\not\in\key{R}$.
By Lemma~\ref{lem:witnessNotFrozen}, $v_{1},\ldots,v_{n}\notin\frozen{q}$.
Then $\tuple{x,\ldots,y}\cdot\tuple{v_{n},v_{n-1},\ldots,v_{1}}$ is a (not necessarily simple) path between a variable of $\vec{x}$ and a variable of $\notkey{R}$ that uses no variable in $\key{R}\cup\frozen{q}$.
It follows that $\vec{x}$ violates condition~\ref{it:separation} of Definition~\ref{def:newclass}, contradicting the lemma's hypothesis. 
\end{proof}

\section{Proofs for Section~\ref{sec:cforest}}

\begin{lemma} \label{lem:pathInFuxmanKeyDetermined}
Let $q(\vec{z})$ be a query in $\cforest$, and let $\tuple{F_{1}, \ldots, F_{n}}$ be a path in the Fuxman graph of $q(\vec{z})$.
Then, $\fdset{q} \models \fd{\key{F_{1}}}{\key{F_{n}}}$.
\end{lemma}
\begin{proof}
By definition of a Fuxman graph, we have that for every $i \in \{1, \ldots, n-1\}$, $\key{F_{i+1}}\setminus\free{q}\subseteq \notkey{F_{i}}$. 
This means that, for every $i \in \{1, \ldots, n-1\}$, $\fdset{q}\models\fd{\key{F_{i}}}{\vars{F_{i+1}}}$.
By Armstrong's transitivity axiom, $\fdset{q} \models \fd{\key{F_{1}}}{\key{F_{n}}}$.
\end{proof}

\begin{lemma}\label{lem:queryPathImpliesFuxmanPath}
Let $q(\vec{z})$ be a query in $\cforest$.
Let $x,y\in\vars{q}$ be distinct variables that are connected in the query graph of~$q$.
Let $R,S$ be two distinct atoms such that $x\in\notkey{R}$ and $y\in\vars{S}$.
Then, the Fuxman graph of $q$ has a directed path from $R$ to $S$.
\end{lemma}
\begin{proof}
We can assume a shortest path $\tuple{v_{1},\ldots,v_{n}}$ in the query graph, where $v_{1}=x$ and $v_{n}=y$.
We can assume a sequence $\tuple{F_{1},\ldots,F_{n+1}}$ of atoms in $\body{q}$ such that:
\begin{itemize}
    \item $F_{1}=R$; 
    \item $F_{n+1}=S$; and
    \item for every $i\in\{1,\ldots,n\}$, $v_{i}\in\vars{F_{i}}\cap\vars{F_{i+1}}$.
\end{itemize}

We will show by induction on increasing $i\in\{1,\ldots,n\}$, that $\tuple{F_{1},\ldots,F_{i+1}}$ is a path in the Fuxman graph of $q$.

\framebox{\emph{Basis $i = 1$.}} 
Since $v_{1}\notin\key{F_{1}}$ and $v_{1}\in\vars{F_{1}}\cap\vars{F_{2}}$, there is a directed edge from $F_{1}$ to $F_{2}$ in the Fuxman graph of $q$.
\framebox{\emph{Induction step $i-1\rightarrow i$.}} 
The induction hypothesis is that $\tuple{F_{1},\ldots,F_{i}}$ is a path in the Fuxman graph of $q$. 
We have that $v_{i}\in\vars{F_{i}}\cap\vars{F_{i+1}}$. 
Assume for the sake of contradiction that $v_{i}\in\key{F_{i}}$.
Then, by definition of $\cforest$, $v_{i}\in\vars{F_{i-1}}$.
Since we also have $v_{i-2}\in\vars{F_{i-1}}$, the query graph of $q(\vec{z})$ contains an edge between $v_{i-2}$ and $v_{i}$.
Then, $\tuple{v_{1},\ldots,v_{i-2},v_{i},\ldots,v_{n}}$ is a shorter path in the query graph between~$v_{1}$ and $v_{n}$, a contradiction.
We conclude by contradiction that $v_{i}\notin\key{F_{i}}$.
It follows that the Fuxman graph of~$q$ has a directed edge from $F_{i}$ to $F_{i+1}$ and, consequently, $\tuple{F_{1},\ldots,F_{i+1}}$ is a path in the Fuxman graph of~$q$.
By induction, we obtain that $\tuple{F_{1}, \ldots, F_{n+1}}$ is a path in the Fuxman graph of $q$ from $R$ to $S$.
\end{proof}

\begin{lemma}\label{lem:attackImpliesFuxmanPath}
Let $q(\vec{z})$ be a query in $\cforest$.
Let $R,S$ be two distinct atoms such that $\attacks{R}{S}{q}$.
Then, the Fuxman graph of $q$ has a directed path from $R$ to $S$.
\end{lemma}
\begin{proof}
We can assume a sequence $\tuple{v_{1},\ldots,v_{n}}$ that is a witness of $\attacks{R}{S}{q}$, where $v_{1}\in\notkey{R}$ and $v_{n}\in\key{S}$.
This sequence is obviously a path in the query graph of $q(\vec{z})$.
The desired result then follows by Lemma~\ref{lem:queryPathImpliesFuxmanPath}.
\end{proof}

\begin{lemma} \label{lem:cForestAcyclicGraph}
Every query in $\cforest$ has an acyclic attack graph.
\end{lemma}
\begin{proof}
Fuxman~\cite{FuxmanThesis} has shown that every query in $\cforest$ has a consistent first-order rewriting. 
From Theorem~\ref{the:koutriswijsen}, it follows that every query in $\cforest$ has an acyclic attack graph.
\end{proof}

\begin{lemma} \label{lem:cForestWeakAttacks}
If $q(\vec{z}) \in \cforest$, then every attack in $q(\vec{z})$'s attack graph is weak.
\end{lemma}
\begin{proof}
Let $q(\vec{z}) \in \cforest$. 
Let $R,S$ be two distinct atoms in $\body{q}$ such that $\attacks{R}{S}{q}$.
By Lemma~\ref{lem:attackImpliesFuxmanPath}, the Fuxman graph of $q$ has a directed path $\tuple{F_{1},\ldots,F_{n}}$ with $F_{1}=R$ and $F_{n}=S$. 
By Lemma~\ref{lem:pathInFuxmanKeyDetermined}, $\fdset{q}\models\fd{\key{R}}{\key{S}}$, and, consequently, the attack from $R$ to $S$ is weak.
\end{proof}


\begin{lemma} \label{lem:cForestGoodComponents}
Let $q(\vec{z}) \in \cforest$. 
Let $\vec{x}$ be a $\subseteq$-minimal tuple of bound variables such that for every root~$R$ in the Fuxman graph of~$q(\vec{z})$, $\vec{x}$ contains every bound variable of $\key{R}$.
Then, $\vec{x}$ verifies conditions~\ref{it:source} and~\ref{it:separation} in Definition~\ref{def:newclass}.
\end{lemma}
\begin{proof}
We first show that $\vec{x}$ verifies condition~\ref{it:source} in Definition~\ref{def:newclass}.
Assume towards a contradiction that $q$'s attack graph has a component $C$ such that $\fdset{q} \not\models \fd{\vec{x}}{\key{U}}$ for every unattacked atom~$U$ that occurs in~$C$.
Since $q$'s attack graph is acyclic by Lemma~\ref{lem:cForestAcyclicGraph}, we can assume an atom~$R$ in the component $C$ that is unattacked.
Consequently, $\fdset{q} \not\models \fd{\vec{x}}{\key{R}}$.
In the Fuxman graph of $q$, let $S$ be the root of the tree that contains~$R$. 
By the hypothesis of the lemma, $\vec{x}$ contains every bound variable of $\key{S}$.
It follows:
\begin{equation}\label{eq:keySdetermined}
\fdset{q}\models\fd{\vec{x}}{\key{S}}.
\end{equation}
If $R=S$, then it follows $\fdset{q}\models\fd{\vec{x}}{\key{R}}$, a contradiction.
Assume next $R\neq S$.
Since~$S$ is the root of the tree in which $R$ appears, there is a path from $S$ to $R$ in the Fuxman graph of $q$. 
By Lemma~\ref{lem:pathInFuxmanKeyDetermined}, we have $\fdset{q}\models\fd{\key{S}}{\key{R}}$.
Using~\eqref{eq:keySdetermined}, we obtain $\fdset{q}\models\fd{\vec{x}}{\key{R}}$, a contradiction.


We next show that $\vec{x}$ verifies condition~\ref{it:separation} in Definition~\ref{def:newclass}.
Assume towards a contradiction that there is an atom $R$ in $\body{q(\vec{z})}$, and a path $\tuple{v_{1}, \ldots, v_{n}}$ in the query graph of $q(\vec{z})$ such that:
\begin{itemize}
    \item $v_{n} \in \vec{x}$; 
    \item $v_{1} \in \notkey{R}\setminus\free{q}$; and
    \item for every $i \in \{1, \ldots, n\}$, $v_{i} \not \in \key{R}$.
\end{itemize}
Since $\vec{x}$ is $\subseteq$-minimal, we can assume an atom  $S$ that is a root in $q$'s Fuxman graph such that $v_{n} \in \key{S}$.
Since $v_{n}\not\in\key{R}$, it follows $R\neq S$.
By Lemma~\ref{lem:queryPathImpliesFuxmanPath}, there is a directed path from $R$ to $S$ in the Fuxman graph of~$q$, contradicting that $S$ is a root.
\end{proof}

\begin{proof}[Proof of Theorem~\ref{the:cforest}]
From Lemmas~\ref{lem:cForestAcyclicGraph}, \ref{lem:cForestWeakAttacks}, and~\ref{lem:cForestGoodComponents}, it follows that $\cforest\subseteq\cnewclass$. 

To show that $\cforest$ is a strict subset, the query $q(z)=\exists x\exists y\, R(\underline{x},y)\land S(\underline{x},y,z)$ is in $\cnewclass$ but not in $\cforest$.
\end{proof}

\end{document}